\documentclass{article}
\usepackage[utf8]{inputenc}

\title{A Hierarchy of Network Models Giving Bistability Under Triadic Closure}

\author{Stefano Di Giovacchino%
    \thanks{%
            Department of Engineering and Computer Science and Mathematics, University of L’Aquila, Italy.
            The work of SDG was 
            supported by the GNCS-INDAM project and by the PRIN2017-MIUR project 2017JYCLSF ``Structure preserving approximation of evolutionary problems.''
           }
\and 
        Desmond John Higham%
            \thanks{%
           School of Mathematics,
           University of Edinburgh,
           Edinburgh, EH9 3FD, UK,
           The work of DJH was supported by
the Engineering and Physical Sciences Research Council 
under grants EP/P020720/1 and EP/V015605/1.
        }
        \and
        Konstantinos C. Zygalakis%
        \thanks{%
           School of Mathematics,
           University of Edinburgh,
           Edinburgh, EH9 3FD, UK.
           The work of KCZ was supported by  the Leverhulme Trust (RF/ 2020-310) and by  the Engineering and Physical Sciences Research Council  under grant EP/V006177/1. 
        }
}

\usepackage{amsmath}
\usepackage{amsmath,amssymb,bbold,verbatim,graphics}
\usepackage{amsmath,amssymb}
\usepackage{amsthm}
\usepackage{subfigure}
\usepackage{multirow}
\usepackage{color}
\newtheorem{theorem}{Theorem}[section]
\usepackage{caption}

\newcommand{\RR}{\mathbb{R}}
\newcommand{\PP}{\mathbb{P}}

\newcommand{\bnu}{\mbox{\boldmath$\nu$}}

\definecolor{raspberry}{rgb}{0.89, 0.04, 0.36}
\definecolor{lincolngreen}{rgb}{0.11, 0.35, 0.02}
\definecolor{mediumtealblue}{rgb}{0.0, 0.33, 0.71}
\definecolor{oceanboatblue}{rgb}{0.0, 0.47, 0.75}
\definecolor{otterbrown}{rgb}{0.4, 0.26, 0.13}
\definecolor{riflegreen}{rgb}{0.25, 0.28, 0.2}
\definecolor{rosewood}{rgb}{0.4, 0.0, 0.04}
\definecolor{uared}{rgb}{0.85, 0.0, 0.3}
\definecolor{vividburgundy}{rgb}{0.62, 0.11, 0.21}
\definecolor{vividauburn}{rgb}{0.58, 0.15, 0.14}

\def\wc3{\widehat{c}_3}

\begin{document}

\maketitle

\begin{abstract}
      Triadic closure describes the tendency for new friendships to form between individuals who already have friends in common.
  It has been argued heuristically that the triadic closure effect 
  can lead to 
bistability in the formation of large-scale social interaction networks.
Here, depending on the initial state and 
the transient dynamics, the system may evolve towards either of two 
long-time states.
 In this work, we propose and study a hierarchy of network evolution models 
 that incorporate triadic closure, building on the work
 of Grindrod, Higham and Parsons 
 [Internet Mathematics, 8, 2012, 402--423].
 We use a chemical kinetics framework, paying careful attention to the reaction rate scaling with respect to the system size. 
In a macroscale regime, we show rigorously that 
 a bimodal steady state distribution is admitted. 
 This behavior corresponds to the existence of two distinct
 stable fixed points in a deterministic mean-field ODE.
 The macroscale model is also seen to capture 
 an apparent metastability property of the microscale system.
 Computational simulations are used to support the analysis. 

\end{abstract}

\section{Motivation and Background}\label{sec:mot}

Network science is built on
the study of pairwise interactions between 
individual elements in a system.
However, it is becoming more widely 
recognized that higher-order motifs
involving 
groups of elements
are also highly relevant
\cite{BRSJK18,benson2016higher,torres2020why}.
In this work we focus on the formation of triangles---connected  
triples of nodes---over time.

Naturally-occurring networks 
are often observed to have an over-abundance of triangles
\cite{KD17,K14}.
In the case of social networks, where
nodes represent people and interactions 
describe friendships, 
\emph{triadic closure}
is widely regarded as a key 
triangle-forming mechanism.
The concept of triadic closure dates back to the work 
of Simmel \cite{Simm08} and gained attention 
after the widely cited article 
\cite{Gran1973}.
Suppose two people, B and C, are not currently friends
but have a friend, A, in common.
The triadic closure principle states that 
the existence of the edges A-B and A-C increases the
likelihood that the edge B-C will form at some future time.
In other words, having a friend in common increases the chance of two people becoming friends.
As discussed in \cite[Chapter~2]{EK2010}, there are three 
reasons why the chance of the B-C connection increases.
First, B and C both socialize with A and hence have 
a greater 
\emph{opportunity} to meet.
Second, A can vouch for B and C and hence 
raise the level of \emph{trust} between them. 
Third, A may be \emph{incentivized} 
to encourage the B-C friendship
if maintaining a single triadic friendship is viewed as more efficient than maintaining a separate pair of dyadic friendships with B and C separately.
Building on these ideas it is also natural, as in 
\cite{cliques20,GHP12}, to argue that 
the increased likelihood that B and C will become 
connected via triadic closure grows in proportion to the number of friends they have in common.

With the advent of large-scale time-stamped data sets 
that record online human interactions, it 
has become possible to test for the presence of
triadic closure 
\cite{BRSJK18,KD17,MH20}.
Moreover, the triadic closure principle can be used
as the basis of link prediction algorithms, which 
attempt to 
anticipate the appearance of new edges over time 
\cite{BRSJK18,EA15,L13,YBU20}.

A time-dependent random graph model that incorporates triadic closure was proposed and analysed in \cite{GHP12}, and also calibrated to real cellphone data.
The model, of the general form defined in
\cite{GH12}, 
takes the form of a discrete-time Markov chain where
edges may appear or disappear independently at random 
over each time step.
Numerical simulations revealed bistable 
dynamics---different paths of the same stochastic process 
were observed to evolve towards either a sparse regime or a rich, well-triangulated regime. 
A heuristic, deterministic mean-field approximation was
put forward to explain the behavior.

In this work, our aim is to develop and study a hierarchy of 
time-dependent network models that incorporate triadic closure and admit bistable behavior.
We use a chemical kinetics framework to describe the system at a microscale level and then consider macroscale approximations.
This allows us to clarify the assumptions that go into 
the micro-to-macro step and to be clear about the scaling of the 
model parameters with respect to system size.
At the macroscale level, the resulting 
process can be shown rigorously to 
admit bistable behavior, in the sense of possessing a bimodal 
steady state distribution for a particular choice of model parameters.
The macroscale steady state also give insights into the 
apparent metastability behavior observed in the microscale
simulations for large system size; where the switching time between sparse and well-triangulated regimes 
becomes extremely long.
We study this effect by 
computing the 
mean time to transition  
between regimes.
We also introduce the 
corresponding diffusion, or Langevin, model and 
compute its 
mean transition times.
Here, reaction rate theory \cite{Pavliotisbook}  
gives a good approximation to the observed growth
in mean passage time with respect to system size.
Finally, we show that the underlying deterministic
reaction rate ODE, arising in the thermodynamic limit,
gives a continuous-time analogue of the 
mean field equation in \cite{GHP12}.

The rest of the manuscript is organized as follows. In section~\ref{sec:micro} we set up the notation and
define the microscale model.
We also report on computational experiments that
motivate the subsequent material.
In section~\ref{sec:macro} we derive a macroscopic approximation
and analyze its steady states and mean exit times.
Section~\ref{sec:sde} introduces the corresponding
Langevin and ODE models.
We finish with a brief discussion in Section~\ref{sec:disc}.

\section{Microscopic Model}\label{sec:micro}

We suppose now that 
a network has a fixed set of $n$ nodes, which
at each point in time may be 
connected by undirected edges.
For example, the 
nodes may represent individuals in an online social media 
platform and the edges may represent mutual friendships.
To be concrete, we assume that the nodes are labelled 
from $1$ to $n$, and we let $N = n(n-1)/2$ denote the total number of edges
that may be present at any time.
We are interested in the case where $n$ is large, say 
$n \ge 30$, and we find it useful to think of 
a \emph{one-parameter family} of models, parameterized by $N$.

Our modelling framework allows edges to be created or deleted in continuous time.
We will use a chemical kinetics setting where 
the creation or deletion of an edge is represented as a reaction
between ``edge'' and ``no edge'' molecules.
To be precise, consider 
any pair of distinct nodes, $i$ and $j$.
Without loss of generality, we take $i < j$.
Then, for the corresponding 
undirected edge between 
nodes $i$ and $j$ we associate two species, 
$E_{ij}$ and $O_{ij}$.
 Existence of $E_{ij}$ represents the presence of the edge  
connecting nodes $i$ and $j$ and existence of
$O_{ij}$ represents the absence of this edge.
So exactly one of
 $E_{ij}$ and $O_{ij}$ exists at any given time 
 $t$.

The mechanisms governing the edge evolution will be modelled as 
the following set of reactions, with rate constants 
$c_1$,
$c_2$ and
$\wc3$:
\begin{eqnarray}
O_{ij} &\stackrel{c_1}{\rightarrow}& E_{ij},    \label{eq:OE} \\
E_{ij}  &\stackrel{c_2}{\rightarrow}& O_{ij},   \label{eq:EO} \\
O_{ij}  + E_{jk} + E_{ik}  &\stackrel{\wc3}{\rightarrow}&  
E_{ij} + 
E_{jk}+E_{ik} \label{eq:EEO} ,
\qquad \text{where~}  i,j,k \text{~are distinct}.
\end{eqnarray}
Here, reactions (\ref{eq:OE})
and 
(\ref{eq:EO})
 represent spontaneous  edge birth and spontaneous edge death,
 respectively. 
 The third reaction, (\ref{eq:EEO}), captures the triadic closure effect:
 if $i$ and $j$ are not currently connected, 
then  
  for every instance where there is a node $k$ connected to both
  $i$ and $j$, there is chance  
  that $i$ and $j$ will become connected via 
  triadic closure.
  Hence, we use the same
  principle as \cite{GHP12}:
  the overall chance of an edge arising from triadic closure is 
  linearly proportional 
  to the number of new triangles that this event would create.
 
 We refer to
(\ref{eq:OE})--(\ref{eq:EEO}) as the \emph{microscopic model}. 
For the purpose of our work, 
this microscopic model is regarded as an exact description of the 
network evolution. Our aim is to study this model, and in 
later sections we 
will introduce approximations that
allow us to gain insights. 
 
 In principle, we could apply the stochastic simulation algorithm (SSA), also known as 
 Gillespie's algorithm, to 
 this system
 \cite{G76,G77}. 
The state vector, which records the number of $E_{ij}$ and $O_{ij}$ ``molecules'' at each time point
 will have $2N$ components.
 There are $N$ reactions of type   (\ref{eq:OE}),
 $N$ reactions of type   (\ref{eq:EO})
and 
$N(N-1)$   reactions of type   (\ref{eq:EEO}).
Hence, the stochiometric vectors would have dimension 
$N(N+1)$. 
Given an initial state vector, 
at each step of the SSA we draw two  
 random numbers: one to determine
time of the next reaction and one to determine 
  which reaction takes place.

Because exactly one of 
$E_{ij}$ and $O_{ij}$ can exist, the propensity functions for 
reactions
(\ref{eq:OE})--(\ref{eq:EEO})  have simple forms:
\begin{itemize}
 \item for reaction (\ref{eq:OE}), the propensity function is $c_1$ if $O_{ij}$ exists and 
 zero otherwise,
  \item for reaction (\ref{eq:EO}), the propensity function is $c_2$ if $E_{ij}$ exists and 
 zero otherwise,
 \item for reaction (\ref{eq:EEO}), the propensity function is $\wc3$ if $O_{ij}$, $E_{jk}$ and $E_{ik}$  exist and 
 zero otherwise.
 \end{itemize}
 
We note that 
reactions 
(\ref{eq:OE}) and (\ref{eq:EO})
involve individual edges, and it is reasonable to 
assume that the rate constants 
$c_1$ and $c_2$ do not depend on the system size.
Reaction (\ref{eq:EEO}), however, involves interactions between edges, and 
the question of how $\wc3$ depends on the system size
could be viewed as 
 a modelling issue.
 We will argue in section~\ref{sec:macro} that 
 $\wc3$ should scale like $1/n$, since this 
 produces 
 systems where  
the triadic closure mechanism  
remains present but does not overwhelm the spontaneous birth and death 
 of edges.
 More precisely, we will let
 \begin{equation}
 \wc3 = \frac{c_3}{n-2},
 \label{eq:c3scale}
 \end{equation}
 for a fixed constant $c_3$ (independent of $n$).


  From the SSA perspective, 
  we can reduce the size 
  of the stoichiometric vectors, and  
   thereby  
   save considerably on 
  storage and bookkeeping, by 
  exploiting the special structure of this system.
  To do this, we    
   draw three 
 random numbers at each step: one to determine
time of the next reaction, one to determine 
  which \emph{class} of reaction takes place and one to 
   pick a reaction from within this class. 
   In this way,  
  the computation can be performed directly with the symmetric adjacency matrix,
  $A(t) \in \RR^{N \times N}$, where $(A(t))_{ij} = 1$ if the edge 
  connecting $i$ and $j$ is present at time $t$ and
   $(A(t))_{ij} = 0$ otherwise.
    (It would be sufficient to work with just the upper triangle of 
     $A(t)$, but we find it more natural to 
     use the whole matrix when describing the algorithm.)
     
   To summarize this approach, we introduce the class-level propensities
   \begin{eqnarray}
         a_{\text{birth}} &=& c_1 \sum_{i<j} \sum_{j} (1 -A(t)_{ij}),   \label{eq:abirth}  \\
    a_{\text{death}} &=& c_2 \sum_{i<j} \sum_{j} A(t)_{ij},    \label{eq:adeath} \\
       a_{\text{triadic}} &= & \wc3 \sum_{i<j} \sum_{j} 
       A_{\text{triadic}}(t)_{ij}, 
         \label{eq:atriadic} 
          \end{eqnarray}
          where 
         $ A_{\text{triadic}}(t) \in \RR^{N \times N}$ is defined as
       \begin{equation}
A_{\text{triadic}}(t)_{ij} =   
(  A^2(t)  )_{ij} (1 -A(t)_{ij}). 
\label{eq:Atriad}
      \end{equation}
In 
 (\ref{eq:abirth}) we form the product of the rate constant $c_1$ and the total number of 
 missing edges.
  Similarly, in  
(\ref{eq:adeath}) we form the product of the rate constant $c_2$ and the total number of 
 current edges. 
 To understand (\ref{eq:atriadic}), note that 
  the term $(  A^2(t)  )_{ij}$ counts the number of nodes that are connected to both 
  $i$ and $j$. 
  Hence, in (\ref{eq:atriadic}) we form the 
product of the rate constant $\wc3$ and the total number of opportunities for 
triadic closure.
  
 A step of SSA may then be performed as follows:
  \begin{description}
   \item[1.] Compute 
       $a_{\text{birth}}$,
      $a_{\text{death}}$,
      $a_{\text{triadic}}$
      and form 
      $a_{\text{sum}} = 
      a_{\text{death}} + 
      a_{\text{birth}} + a_{\text{triadic}}$.
        \item[2.] Choose the time until the next reaction from an exponential distribution with parameter
          $a_{\text{sum}}$.
           \item[3.]    Choose the class of reaction: death, birth or triadic, with probability proportional to 
           $a_{\text{death}}$,
      $a_{\text{birth}}$ and
      $a_{\text{triadic}}$, respectively.
      \item[4a.] If the class is death, choose one of the nonzeros from the upper triangle of $A(t)$ uniformly at random, and set it to zero.
          \item[4b.] If the class is birth, choose one of the zeros from the upper triangle of $A(t)$ uniformly at random, and set it to one.
            \item[4c.] If the class is triadic, choose one of the nonzeros from the upper triangle of the matrix 
             $A_{\text{triadic}}(t)_{ij}$ with probability proportional to its value, and set it to one.
        \end{description}

\begin{figure}[h]
  \centering
		\subfigure{\scalebox{0.24}{\includegraphics{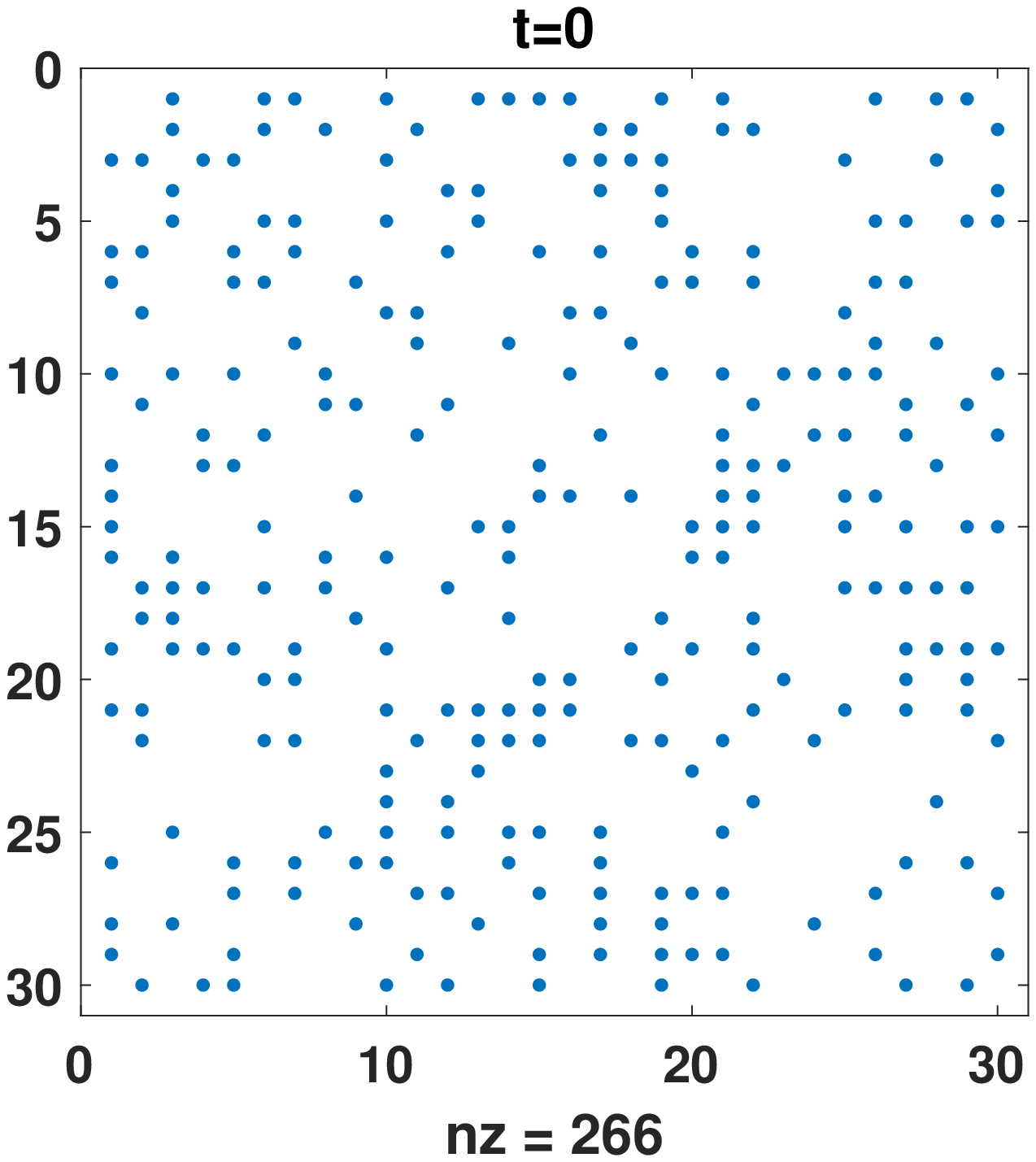}}}\quad
		\subfigure{\scalebox{0.24}{\includegraphics{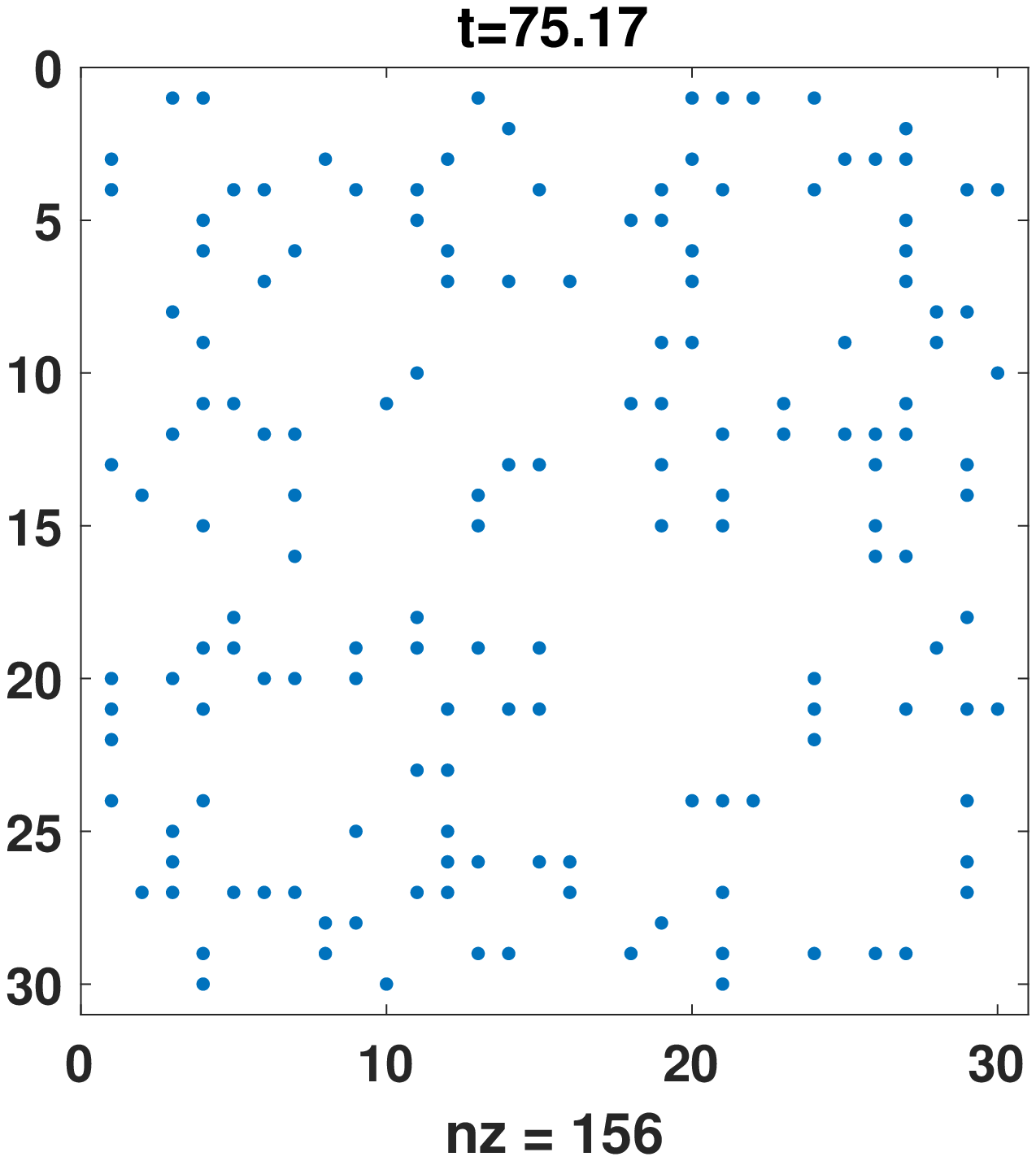}}}\\
		\subfigure{\scalebox{0.24}{\includegraphics{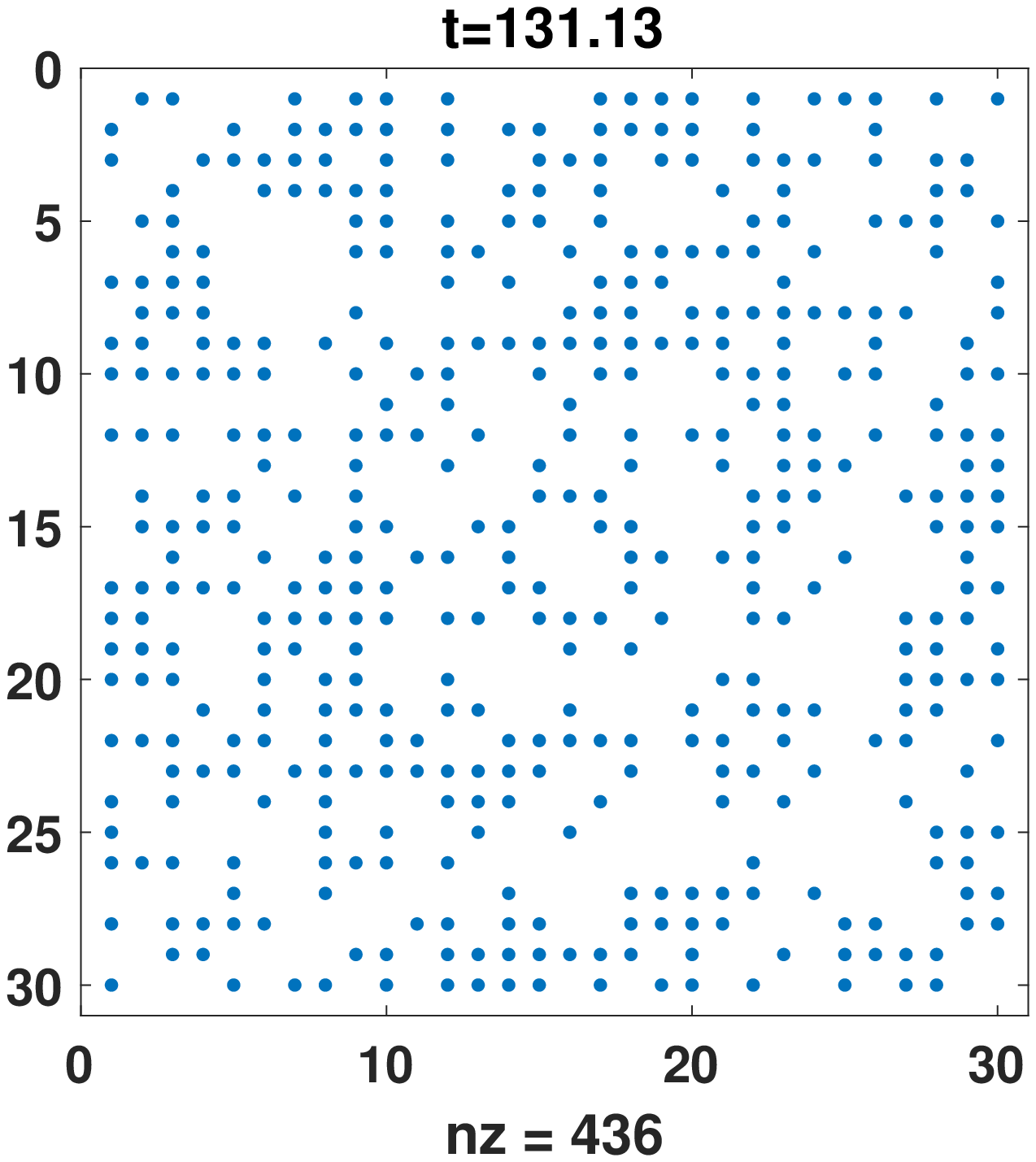}}}\quad
		\subfigure{\scalebox{0.24}{\includegraphics{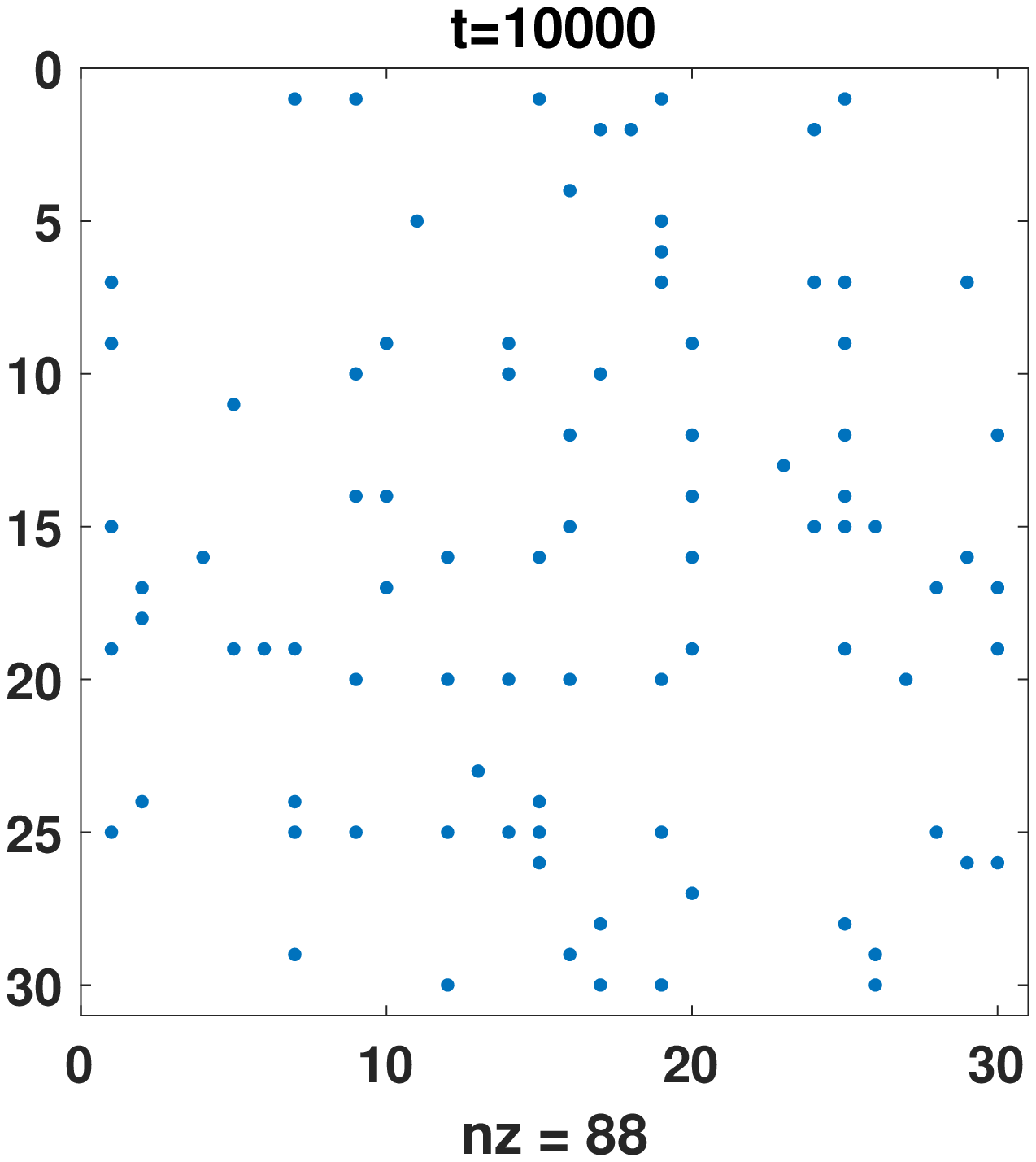}}}	
 \caption{Time evolution of the adjacency matrix $A(t)$. Nonzeros are represented by dots.
 We show the connectivity at the initial time (upper left), at time of low density (upper right), a time of 
 high density (lower left) and a later time
where the density is again low (lower right).
 Here, we have $n=30$ nodes and \texttt{nnz} denotes the number of nonzeros in the adjacency matrix.}
  \label{fig:spy}
\end{figure}
 
It is worth pointing out here that the elements of the adjacency matrix  
represent chemical species, and they can only take the value $0$ or $1$. However, given that there exist $N$ of them (using the fact that the adjacency  matrix is symmetric) the Markov chain simulated by the algorithm described above has a state space of dimension $2^{N}$ (all the possible symmetric graphs with $n$-nodes). This in turn implies that as the number of nodes $n$ increases the computational cost of the algorithm will become a 
severe limiting factor in understanding the dynamics of the system. One possible way of accelerating the simulation is to use the $\tau$-leaping algorithm \cite{G2001tau}. 
In the case $\tau = 1$, this would essentially 
reduce to the model in \cite{GHP12}. 

We visualise in Figure~\ref{fig:spy} the adjacency matrix $A(t)$ at different time instances
along a path computed by the SSA algorithm described above. Here a dot in the picture corresponds to the existence of an edge. 
We used reaction rate coefficients
\begin{equation}
c_1 = 0.025, \quad c_2 = 0.25,  \quad c_3 = 0.91,
\label{eq:bistablec}
\end{equation}
with $\wc3$ defined in (\ref{eq:c3scale}).
Here, the initial configuration was a sample from
the Erd\H{o}s-R\'enyi
class with parameter $p=0.3$, i.e.,  every 
edge has independent probability $p$ 
to exist at $t=0$.
As $t$ increases the system is seen to transition 
between states of low and high connectivity.  
We can summarize the overall connectivity by the edge density
\begin{equation} \label{eq:macro_q}
q(t) := \frac{1}{N}\sum_{i<j} \sum_{j} A(t)_{ij}.    
\end{equation}
In Figure~\ref{fig:spy}, the upper and lower right 
states have edge density of around
$0.18$ and $0.1$, respectively, and the lower left state has edge density of around $0.5$.
The transitioning becomes clearer 
in Figure~\ref{fig:q}, where 
we plot the evolution of $q(t)$
for a range of different system sizes,
again using an 
Erd\H{o}s-R\'enyi
initial condition with 
 $p=0.3$. 
As $n$ is increased from $30$ to $50$, the 
edge density spends longer periods of time around each level. 
For $n= 80$ and $n=100$, switching does not take place
over the interval $0 \le t \le 10^4$.
 We will return to these observations in the next section, where we construct and study an explanatory macroscopic 
 approximation.

\begin{figure}[h]
  \centering
		\subfigure{\scalebox{0.24}{\includegraphics{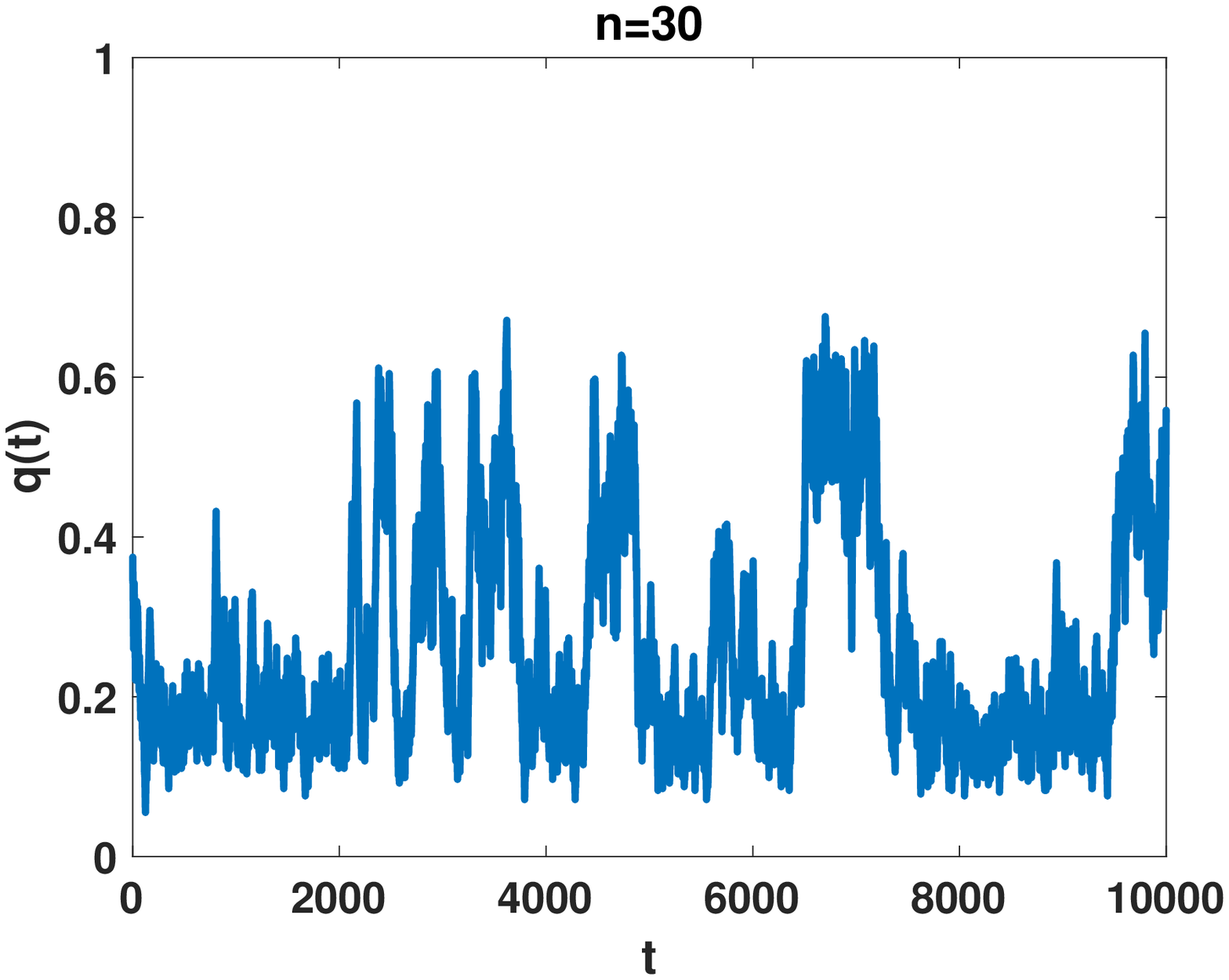}}}\quad
		\subfigure{\scalebox{0.24}{\includegraphics{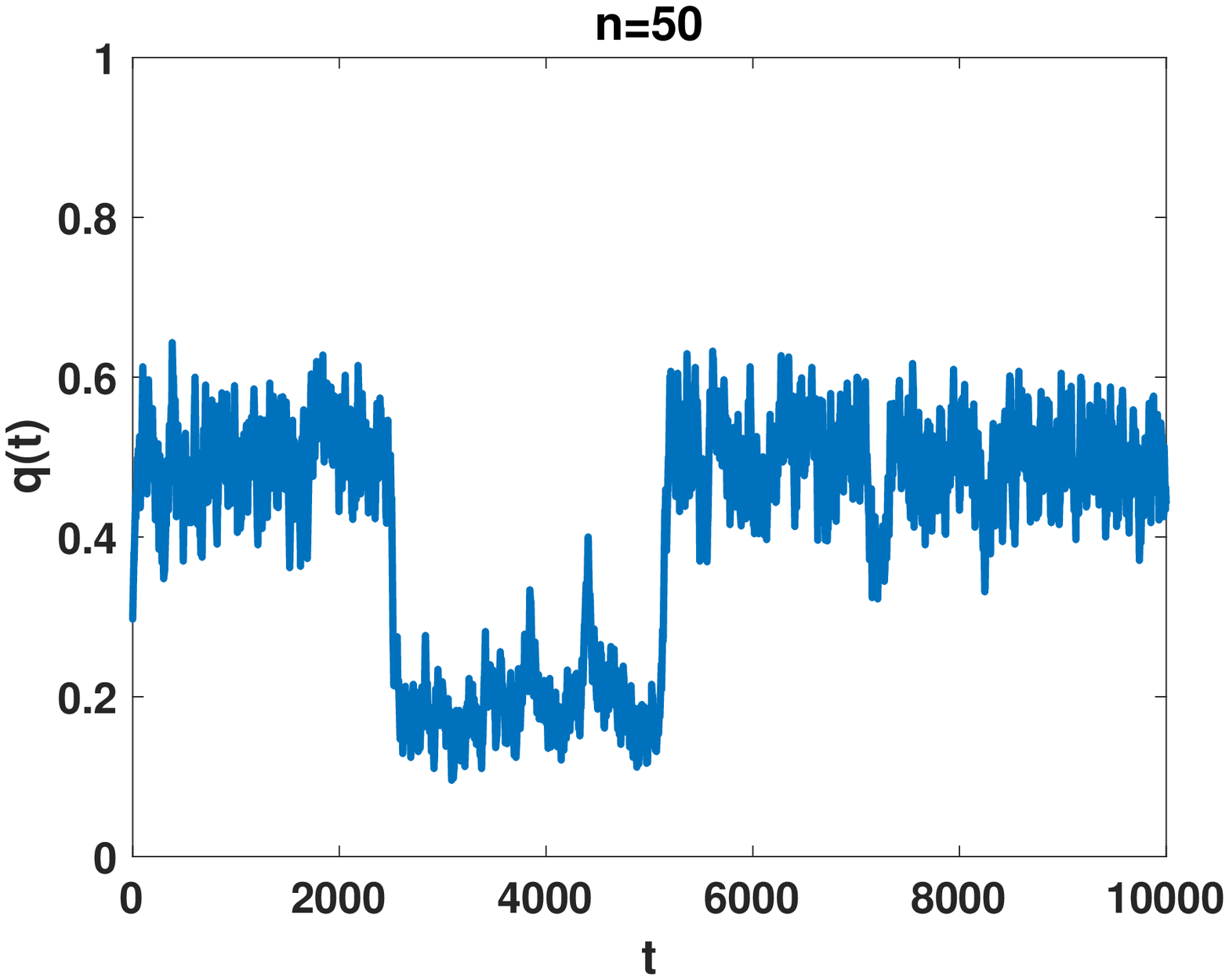}}}\quad
		\subfigure{\scalebox{0.24}{\includegraphics{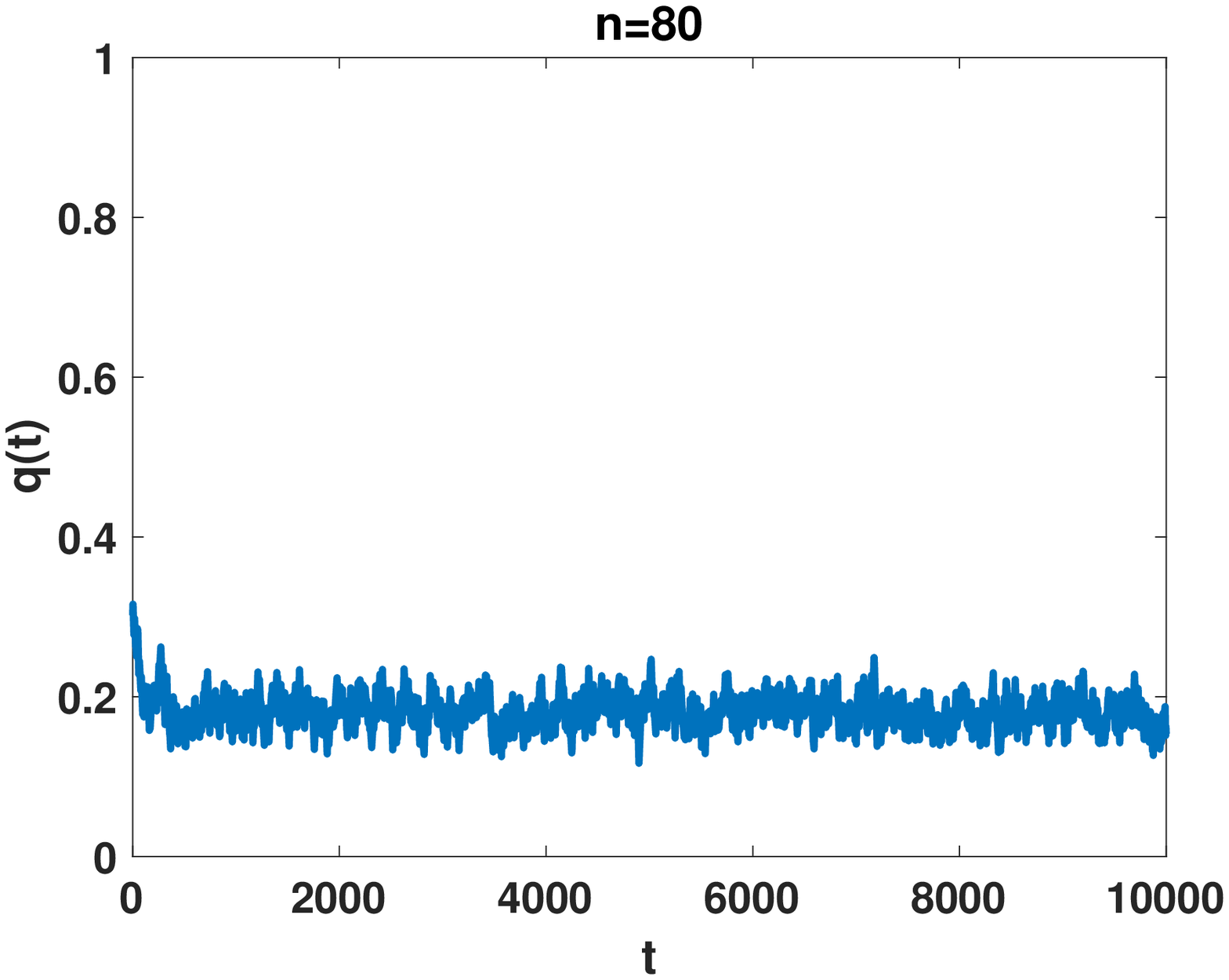}}}\quad
		\subfigure{\scalebox{0.24}{\includegraphics{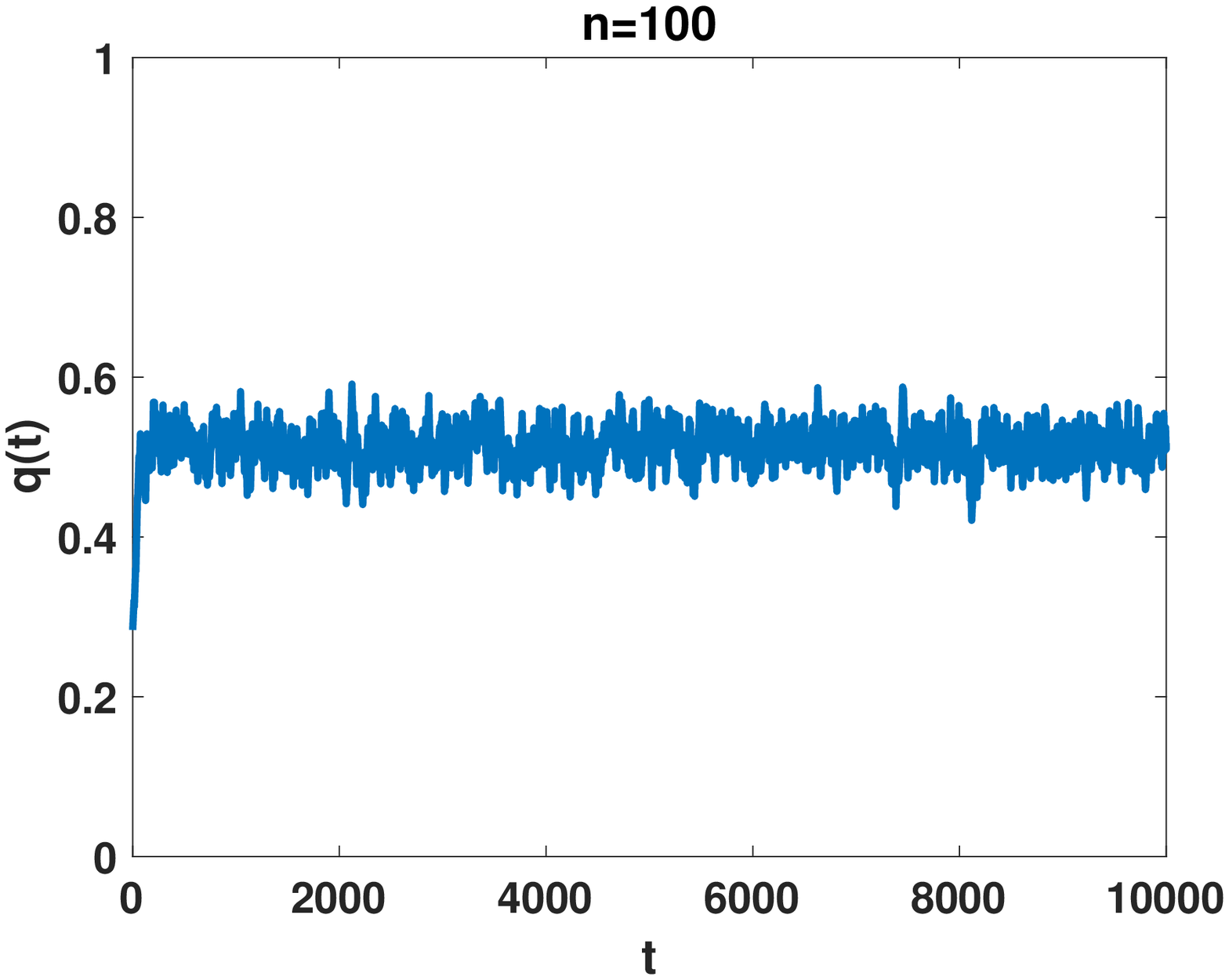}}}
 \caption{Evolution of edge density for $n=30,50,80,100$, up to time $t=10^4$.}
  \label{fig:q}
\end{figure}

\begin{figure}[htbp]
\centering
\scalebox{0.4}{
\includegraphics{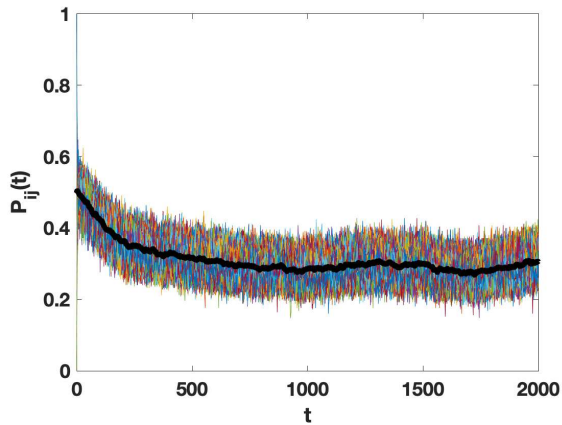}}
\caption{Time evolution of all $P_{ij}(t)$, the 
probability that the edge $i \leftrightarrow j$ exists at time $t$. Thick black line is the mean of $P_{ij}(t)$ over all
$i$ and $j$.}
\label{fig:pij}
\end{figure}

Given that the overall edge density is a macroscopic quantity, it is natural to ask if it captures the effective behavior of the system. To investigate this
issue, let $P_{ij}(t)$ denote the probability 
that the edge connecting $i$ and $j$ exists at time $t$.
In Figure~\ref{fig:pij} we superimpose all the individual values $P_{ij}(t)$ for the case $n=30$.
Here, we used a Monte Carlo approach to estimate each 
$P_{ij}(t)$ by applying SSA to compute $250$
independent paths. For each path, an initial state was chosen 
where half of the edges exist.
We see that after a fast transient all probabilities fluctuate around a value that 
lies between the low and high edge density regimes 
that were seen in Figures~\ref{fig:spy} and \ref{fig:q}.
(The mean of $P_{ij}(t)$ over all $i$ and $j$ is 
superimposed as a thick black line.) 
This implies that after the initial transient the topology of the network is homogeneous in the sense that 
no nodes have distinguishable behavior.







\section{Macroscopic Model}\label{sec:macro}

We saw in section~\ref{sec:micro} that 
 after an initial transient the specific topology of the network appears to be homogeneous. 
 Motivated by this observation we now define a simplified stochastic model. 
 In this \emph{macroscopic} model we do not retain any information about the 
topological structure of the network; we simply keep track of the number of edges.
This reduces the size of the state space dramatically;
in the microscopic model there 
are $2^N$ distinct networks. In this macroscopic model 
the scalar edge count ranges between $0$ and $N$. 
The simpler model allows long-time numerical simulation
to be performed more efficiently, and, as we show in subsection~\ref{subsec:steady}, it is also 
amenable to rigorous analysis. 

\subsection{Macroscopic Regime}

Our simplified, macroscopic version 
 of the microscale system 
(\ref{eq:OE})--(\ref{eq:EEO})
involves two species: 
$E$ represents an edge and 
$O$ represents a missing edge.
The species undergo three types of reaction, representing birth, death and triadic closure respectively:
\begin{eqnarray}
O &\rightarrow& E,    \label{eq:OEm} \\
E  &\rightarrow& O,   \label{eq:EOm} \\
O  + E + E  &\rightarrow&  
E + 
E+ E \label{eq:EEOm}.
\end{eqnarray}
The state vector may be written $X(t) \in \RR^2$, where 
$X_1(t)$ and $X_2(t)$ record the number of edges (species $E$) and missing edges (species $O$), respectively. Notice that $X_{1}(t)+X_{2}(t)=N$ for all $t\geq 0$.
The stoichiometric vectors for the three reactions are 
\[
\bnu_1 =
      \left[
           \begin{array}{r}
                 1 \\
                 -1 
           \end{array}
   \right],
   \quad
   \bnu_2 =
      \left[
           \begin{array}{r}
                 -1 \\
                 1 
           \end{array}
   \right],
   \quad
     \bnu_3 =
      \left[
           \begin{array}{r}
                 1 \\
                 -1 
           \end{array}
   \right].
\]

Reactions (\ref{eq:OEm}) and (\ref{eq:EOm}) are first order reactions that model spontaneous birth and death of edges.
These effects are independent of the network structure, and hence we reproduce the corresponding birth and death 
behavior from the microscale model exactly by taking propensity functions
$a_1(X(t)) = c_1 X_2(t)$, and
 $a_2(X(t)) = c_2 X_1(t)$, where $c_1$ and $c_2$ are the rate constants from 
 (\ref{eq:OE}) and 
 (\ref{eq:EO}).
 For the reaction (\ref{eq:EEOm}) representing triadic closure, we
 must introduce simplifying assumptions, because 
 the macroscale regime does not keep track of 
 individual node and edge labels. 
 Based on our observation in section~\ref{sec:micro}, we suppose 
 that at the current time $t$ edges 
 are uniformly and independently 
 distributed among the nodes---every possible 
 edge $i \leftrightarrow j$ (with $i<j$) has 
a chance $X_1(t)/N$ of existing.
 (We note that this is effectively the mean field assumption used in 
\cite{GHP12}.)
Our aim is now to approximate the number of 
\emph{open wedges}; that is, node triples $i,j,k$ such that edges $i \leftrightarrow j$ and $i \leftrightarrow k$ are present and edge $j \leftrightarrow k$ is absent. This
quantity records the number of opportunities for the 
microscale
reaction (\ref{eq:EEO}) to take place, and hence will 
be used in forming the propensity function at the macroscale level
in (\ref{eq:EEOm}).

To derive this approximation,
we first note there are $X_1(t)$
ways to choose an edge $i \leftrightarrow j$.
Then, fixing this $i$ and $j$, 
for every other node $k$
there is a chance of 
$\approx (X_1(t)-1)/N$ that the edge $i \leftrightarrow k$ 
also exists.
So, taking into account all $n-2$ nodes, there are 
 $\approx (n-2) X_1(t) (X_1(t) - 1)/N$ 
chains of the form $j \leftrightarrow i \leftrightarrow k$.
For any such chain, the chance of the third edge, 
$j \leftrightarrow k$,
being absent is $\approx X_2(t)/N$. So the overall
number of open wedges is 
$\approx (n-2) X_1(t) (X_1(t) - 1)  X_2(t)/N^2$.
 Hence, we find that a suitable propensity function for reaction 
  (\ref{eq:EEOm}) is
  \begin{equation}
  a_3(X(t)) = \wc3 (n-2) X_2(t) X_1(t) (X_1(t) - 1) /N^2,
  \label{eq:a3prop}
  \end{equation}
  where $\wc3$ is the rate constant 
  in 
 (\ref{eq:EEO}).
 
Since $X_1(t) + X_2(t) = N$, the 
macroscopic model
(\ref{eq:OEm})--(\ref{eq:EEOm})
may be written as a scalar, continuous-time birth and death 
process for which $X_1(t)$ takes integer values between $0$ and $N$.
Using $\lambda_i$ and $\mu_i$ to denote the overall birth and death rates, we have the general form 
\begin{eqnarray}
\PP  \left( X_1(t+\delta t) = i+1 | X(t) = i \right) &=& 
\lambda_i \, \delta t + o(\delta t),
  \text{~for~} i = 0, \ldots, N-1,  \label{eq:bdb}\\
  \PP  \left(( X_1(t+\delta t) = i-1 | X(t) = i \right) 
  &=& \mu_i  \, \delta t + o(\delta t),
  \text{~for~} i = 1, \ldots, N.
  \label{eq:bdd}
\end{eqnarray}
Based on the arguments above, these birth and death rates take the form 
\begin{eqnarray}
\lambda_i &=& c_1 (N-i) + \frac{\wc3 \, (n-2)}{N^2} (N-i) i (i-1), 
 \label{eq:li}\\
 \mu_i &=& c_2 \, i. 
 \label{eq:mi}
\end{eqnarray}
Using the scaling (\ref{eq:c3scale}),  
we obtain
the birth rate
\begin{equation}
  \lambda_i = N \left[ c_1 \left( 1 - \frac{i}{N} \right) 
       + c_3 \left(1 - \frac{i}{N}\right)
          \left(\frac{i}{N} \right)
           \left(\frac{i-1}{N} \right)
          \right].
\label{eq:new_lambda_i}
\end{equation}
We then rewrite the death rate (\ref{eq:mi}) as
\begin{equation}
 \mu_i = N c_2 \, \frac{i}{N}. 
 \label{eq:new_mu_i}
\end{equation}
We note in particular that the 
$\wc3 \propto 1/\sqrt{N}$
scaling 
from (\ref{eq:c3scale})  
has balanced the size of the 
separate terms in the
birth and death rates
(\ref{eq:new_lambda_i}) and (\ref{eq:new_mu_i}).

We emphasize that 
$c_1$, $c_2$ and $c_3$ are regarded as positive constants
and we have a family of 
birth and death processes
parametrized by the system size, $N$, representing the number of edges.

\subsection{Comparison between the Macroscopic and Microscopic Models}

Having derived an approximation to the 
microscopic model in the previous 
subsection it is 
natural to compare individual paths of 
the two stochastic models. We will choose two sets of reaction constants $c_{1},c_{2},c_{3}$.  
From now on, and for reasons that will become apparent, we will refer to the parameters in
(\ref{eq:bistablec}) as the bistable regime, 
and we will also use a second set of parameters with a larger birth rate, $c_1$, given by  
\begin{equation}
c_1 = 0.25, \quad c_2 = 0.25,  \quad c_3 = 0.91.
\label{eq:monostablec}
\end{equation}
We will refer to (\ref{eq:monostablec}) as the monostable regime.

The behavior of paths in the monostable regime can be seen  in the upper plots of Figure~\ref{fig:path_comparison}, for 
$n=30$ on the left and $n=100$ on the right. 
The initial state 
was an 
Erd\H{o}s-R\'enyi
sample with parameter $p=0.2$.
For both values of $n$, we see that the microscopic and macroscopic models are behaving in the same way, namely having the edge density increasing towards a high value and oscillating around it. Additionally, we superimpose in these plots the solution of the
deterministic mean field model in \cite{GHP12} (see also the discussion in Section \ref{sec:sde}). As we can see, both models agree closely with the mean field approximation, with the size of fluctuations becoming smaller for larger $n$.

\begin{figure}[h]
  \centering
		\subfigure{\scalebox{0.2}{\includegraphics{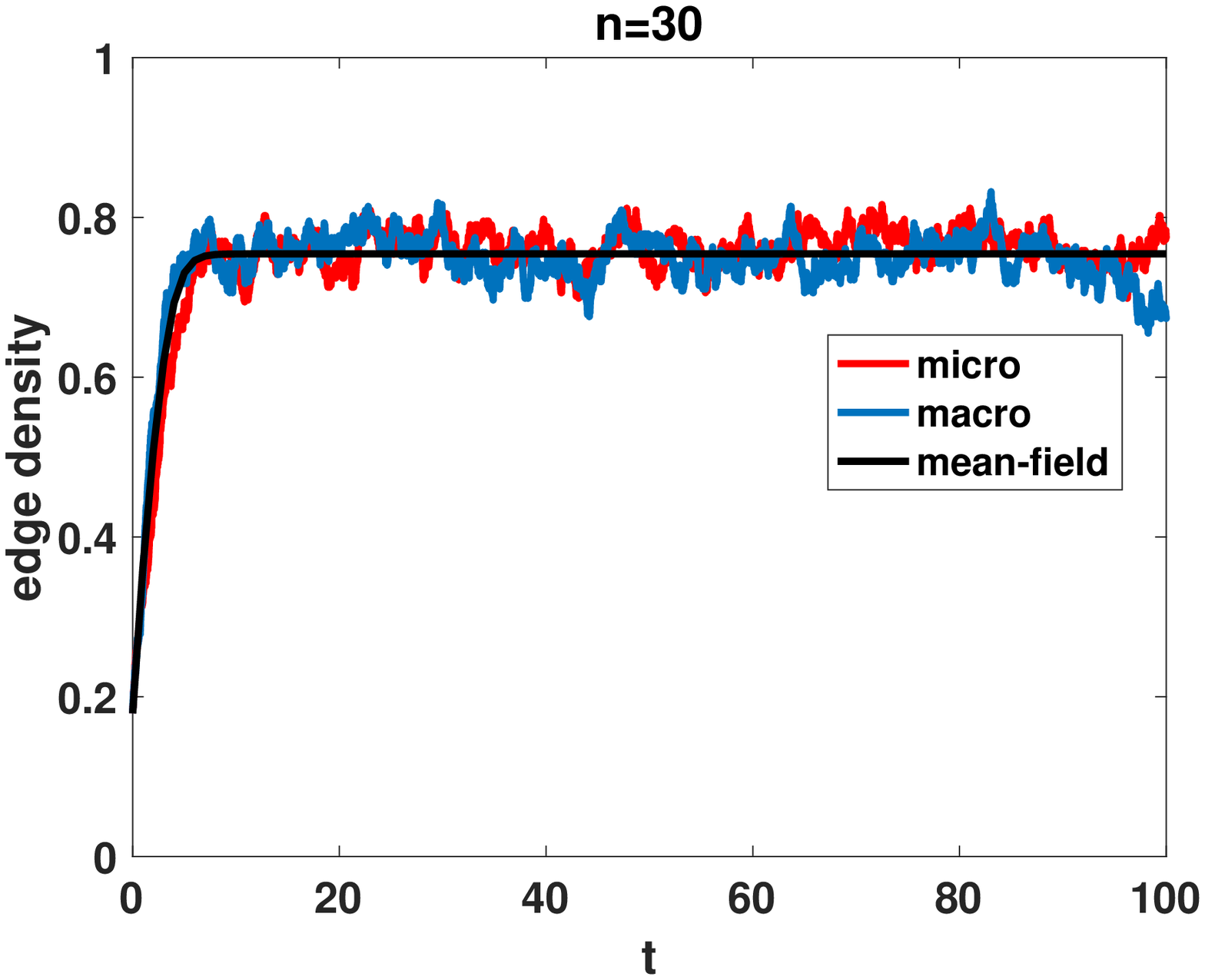}}}\quad
		\subfigure{\scalebox{0.2}{\includegraphics{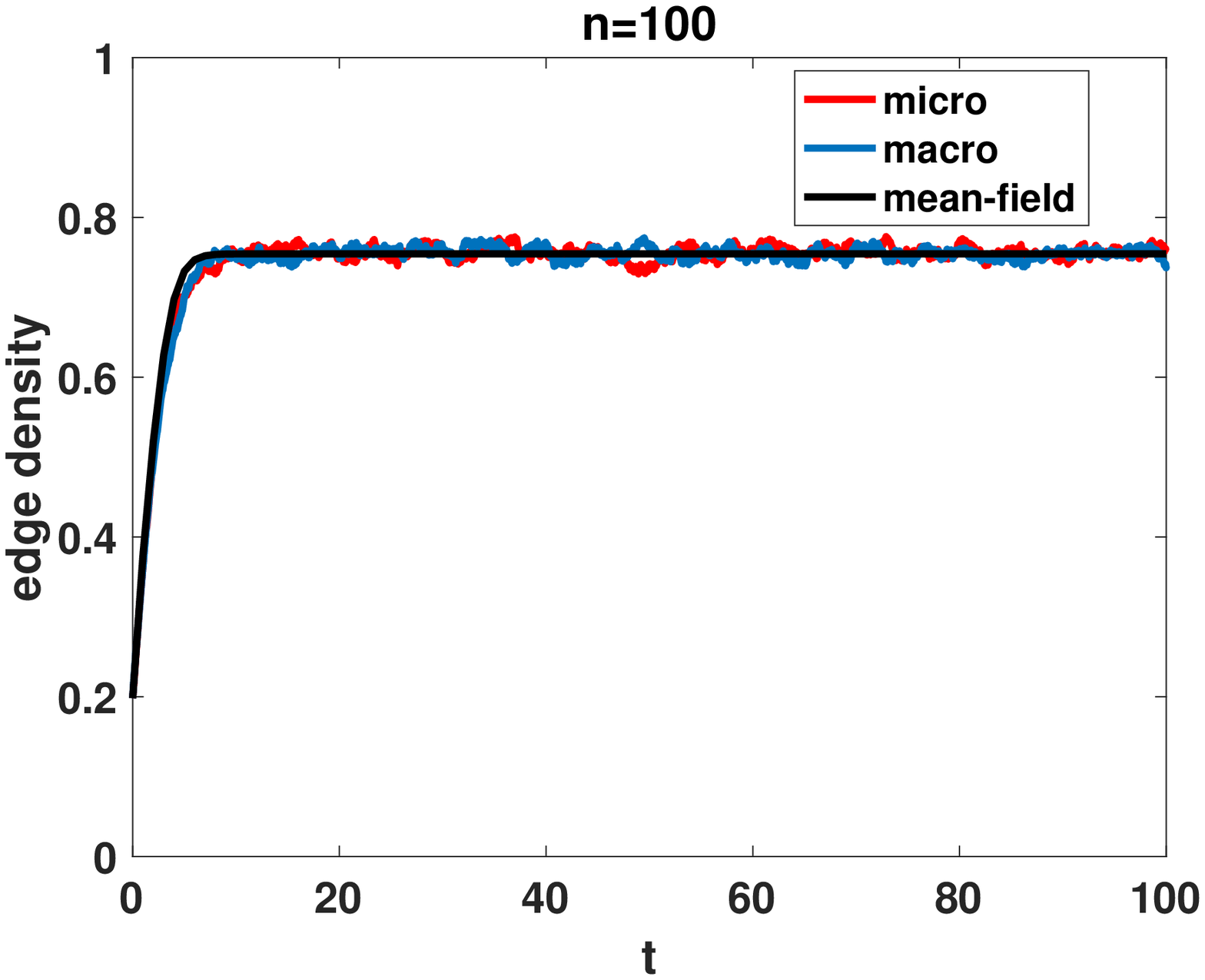}}}\\
		\subfigure{\scalebox{0.2}{\includegraphics{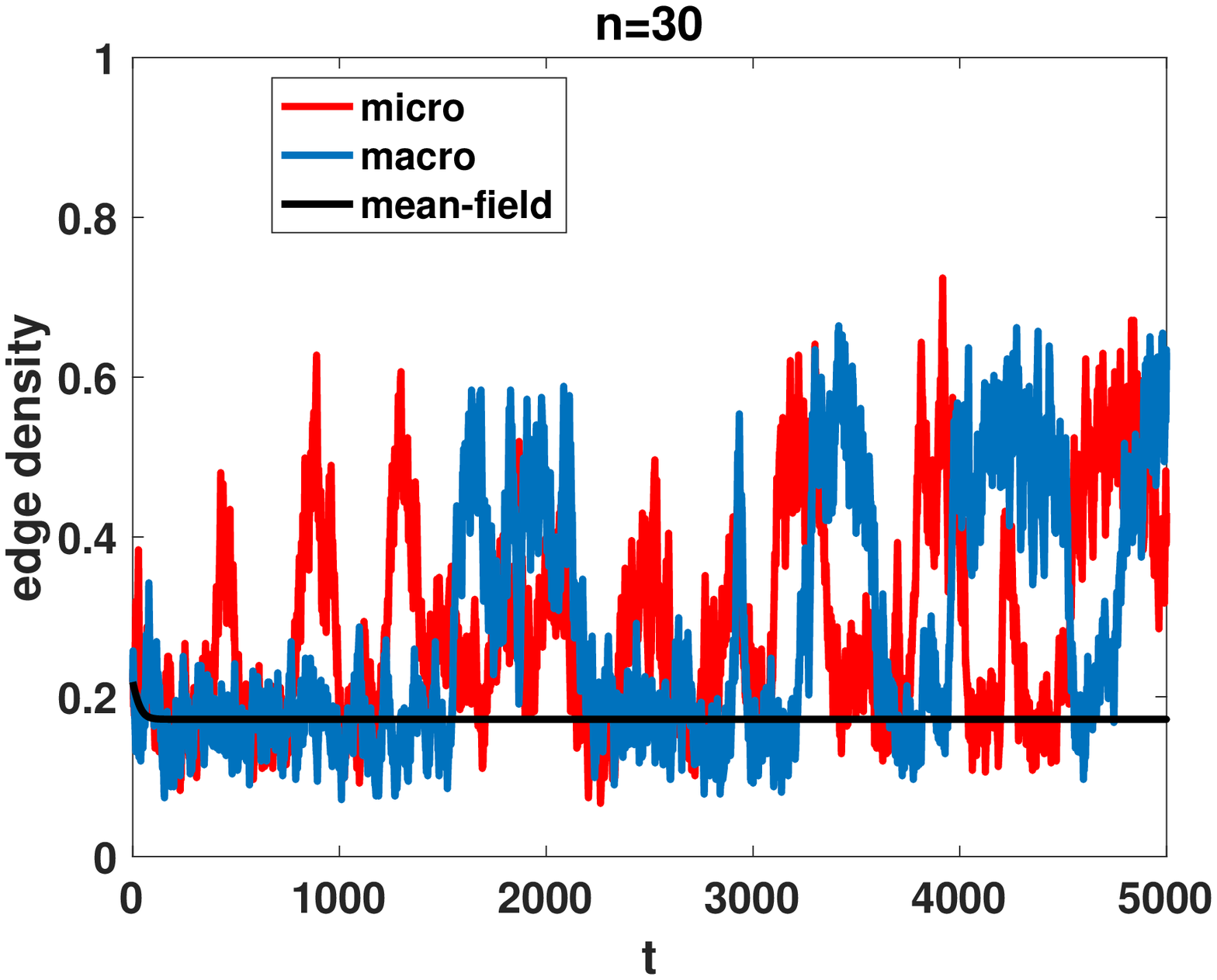}}}\quad
		\subfigure{\scalebox{0.2}{\includegraphics{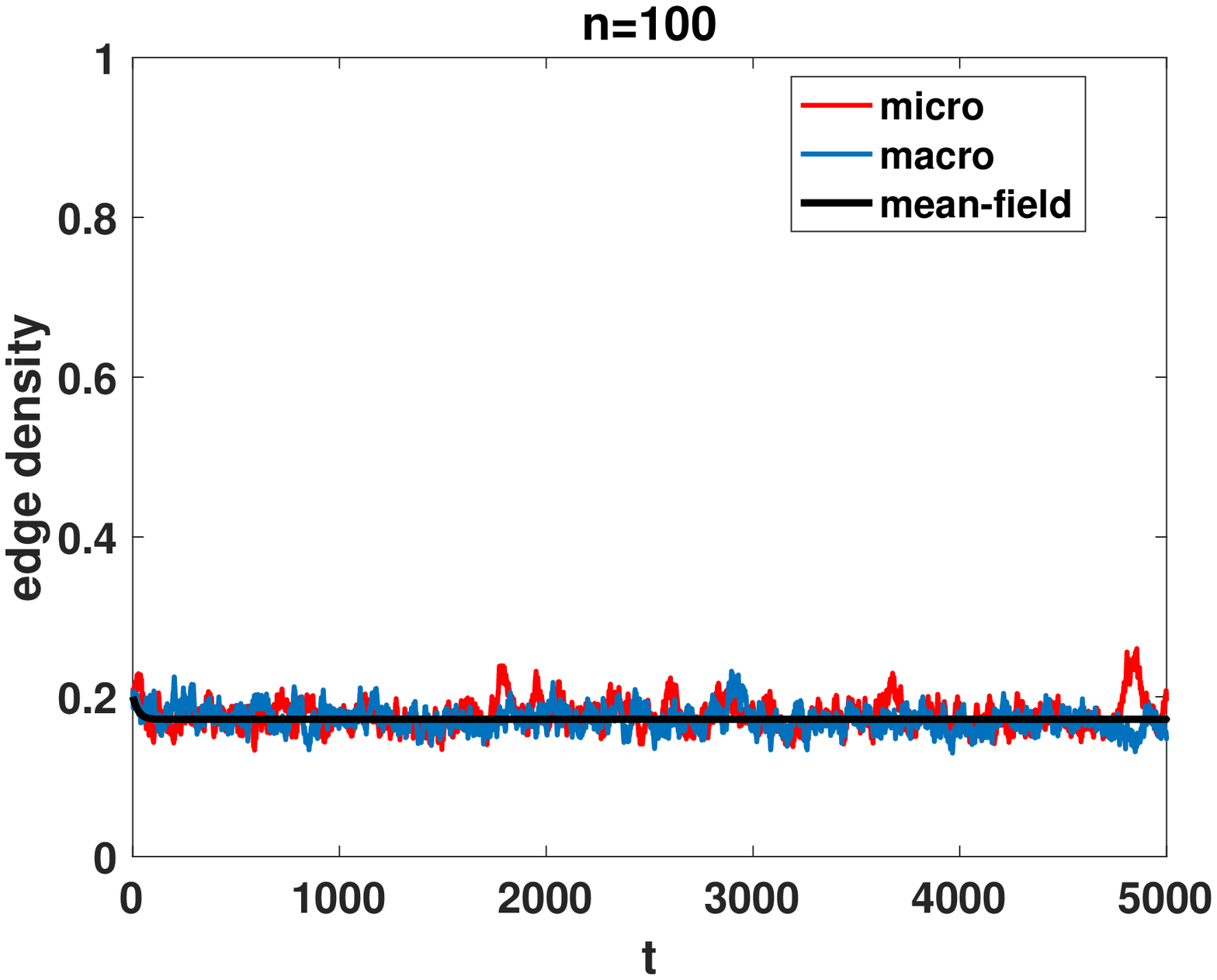}}}	
 \caption{Comparison between the microscopic, the macroscopic and the mean field model. 
 Erd\H{o}s-R\'enyi initial configuration with 
 $p = 0.2$ using 
 $n=30$ (left) and $n=100$ (right).
 Upper plots:
 monostable regime (\ref{eq:monostablec}).
 Lower plots:
 bistable regime (\ref{eq:bistablec}).}
  \label{fig:path_comparison}
\end{figure}
 
The behavior of the paths in the bistable 
parameter regime  (\ref{eq:bistablec}) can be seen in the lower plots of 
Figure~\ref{fig:path_comparison}. 
Again we have $n=30$ on the left and $n=100$ on the right and we used an 
Erd\H{o}s-R\'enyi initial configuration with $p=0.2$. 
We see again there is excellent qualitative agreement between the microscopic and macroscopic paths. However, unlike the monostable regime,
for $n = 30$ the paths now do not agree with the mean field model,
which, by its deterministic nature, 
does not switch between the low and high density regimes. 

\begin{figure}[h]
  \centering
		\subfigure{\scalebox{0.29}{\includegraphics{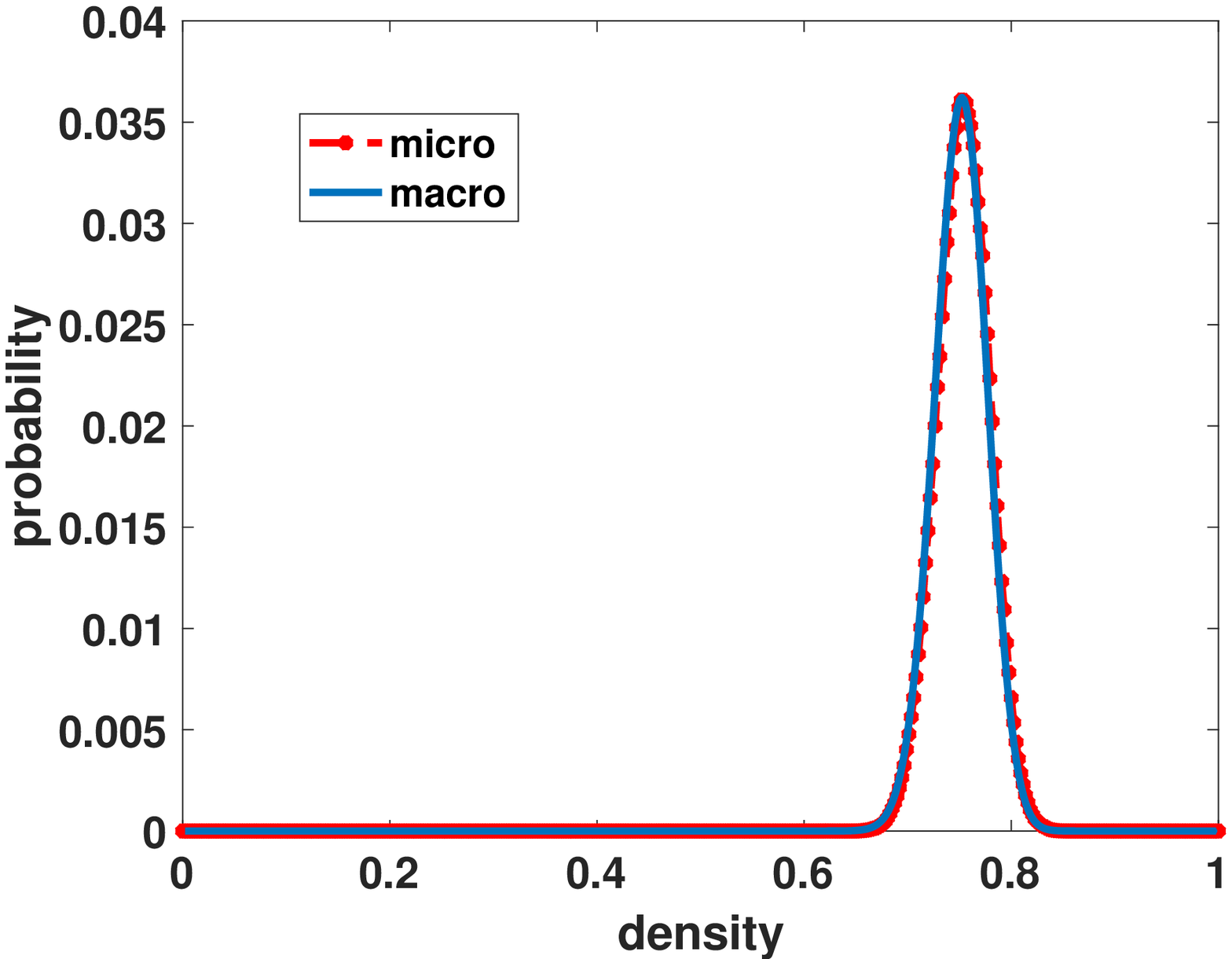}}}\quad
		\subfigure{\scalebox{0.29}{\includegraphics{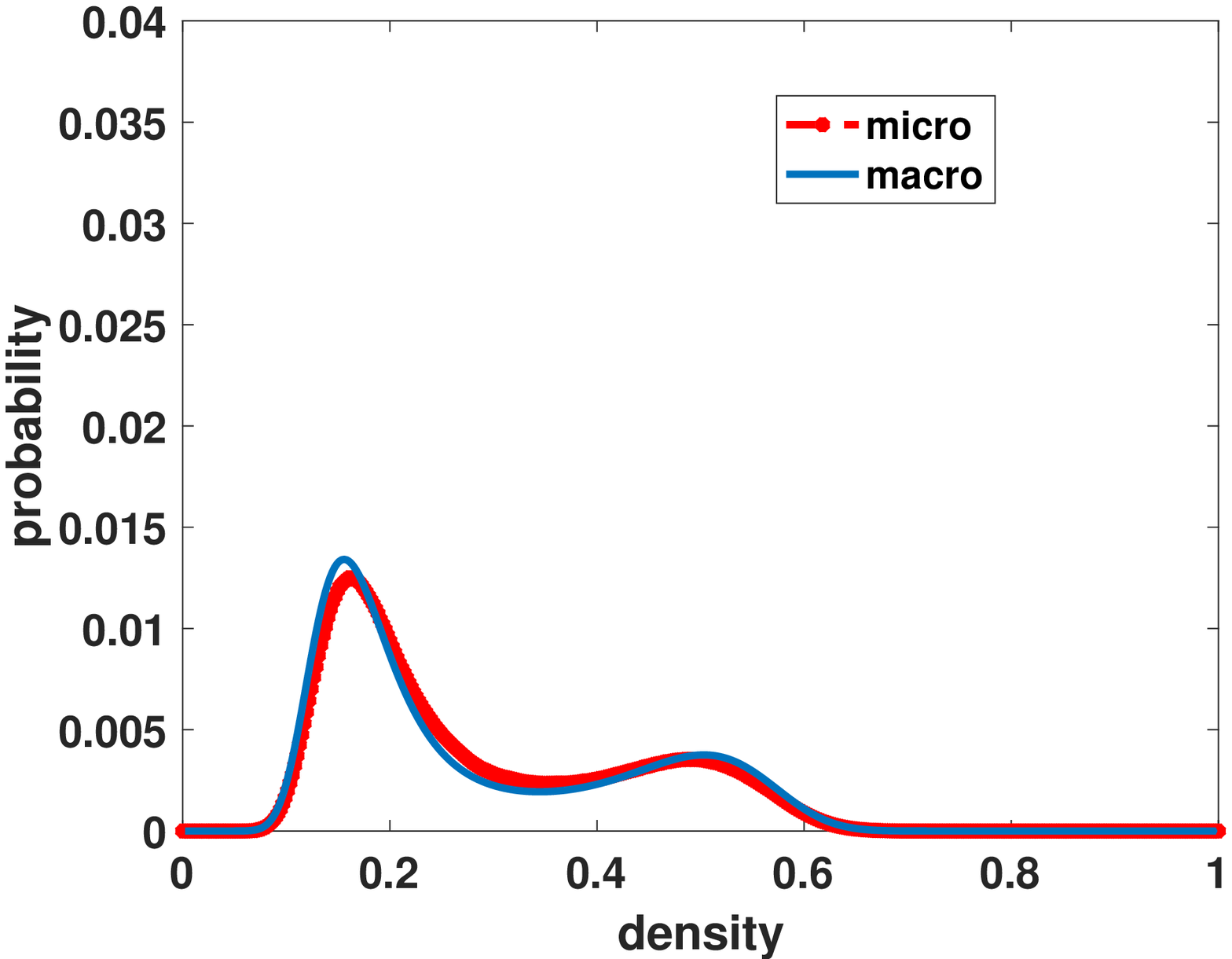}}}\\
	    \caption{Comparison between the steady states of the microscopic and macroscopic models in the monostable 
	    (left, parameter set (\ref{eq:monostablec})) and bistable 
	    (right, parameter set (\ref{eq:bistablec})) 
	    regimes for $n = 30$.}
  \label{fig:steady_state_comparison}
\end{figure}

Our second comparison of the two models concerns 
their long time behavior. In the case of 
the macroscopic model, we will see in the next 
section that it is possible to obtain an analytic 
expression for the steady state.
For the microscopic 
model it is not  possible to derive an analytic expression for the 
steady state, so we used  
a long-time simulation ($T=5 \times 10^6$ in the bistable case, $T=10^4$ in the monostable case). 
We plot this comparison in Figure \ref{fig:steady_state_comparison}.
Here, the horizontal axis represents the scaled state values $0,1/N,2/N,\ldots,1$.  
We observe excellent agreement between the two steady states both in the monostable and bistable regime.

\subsection{Steady State of Macroscopic Model}\label{subsec:steady}

We now show that the bistable behavior 
observed in Figures~\ref{fig:spy} and \ref{fig:path_comparison} 
can be explained by analysing the 
long term behavior of the macroscopic model.
Standard theory \cite{TaylKarl98} shows that in our ergodic
setting a process of the form 
(\ref{eq:bdb})--(\ref{eq:bdd}) 
has a stationary distribution 
\[
\pi_j =  \mathbb{P}\left( \lim_{t\to \infty} X_1(t) = j\right), \quad j=0, 1,\ldots, N, 
\]
that satisfies 
\begin{equation}
\label{eq:exact_steady}
\pi_{j} = \displaystyle\frac{\prod_{i=0}^{j-1}\lambda_i}{\prod_{i=1}^{j}\mu_j}\pi_0.
\end{equation}

In Figure~\ref{fig:steady} we show the steady state distribution 
in the 
bistable regime (\ref{eq:bistablec})
when the number 
of nodes is $n = 30$, $50$, $80$, and $100$.
This figure may be compared with Figure~\ref{fig:q}. 
For internal consistency 
we have regarded these four curves
as continuous-valued probability density functions
over $[0,1]$, thereby normalizing them to 
have unit area. 
We see 
in 
Figure~\ref{fig:steady}
that the bimodal steady state distribution has peaks at around 
$0.2$ and $0.5$, in agreement with 
the low and high connectivity regimes seen in   
Figure~\ref{fig:q}.
We also see that although the location of the 
two peaks seems to be fixed, their relative heights vary.
For $n=30$, the lower mode in 
Figure~\ref{fig:steady}
captures more probability,
which is consistent with the trajectory of $q(t)$ in  
Figure~\ref{fig:q} spending more time in the lower
connectivity regime. As $n$ increases the higher connectivity regime begins to dominate, and indeed the 
bimodality is barely visible for $n=100$ in 
Figure~\ref{fig:steady}.  This is consistent with the less frequent switching  
observed in Figure~\ref{fig:q}.


\begin{figure}[h]
  \centering
		\subfigure{\scalebox{0.24}{\includegraphics{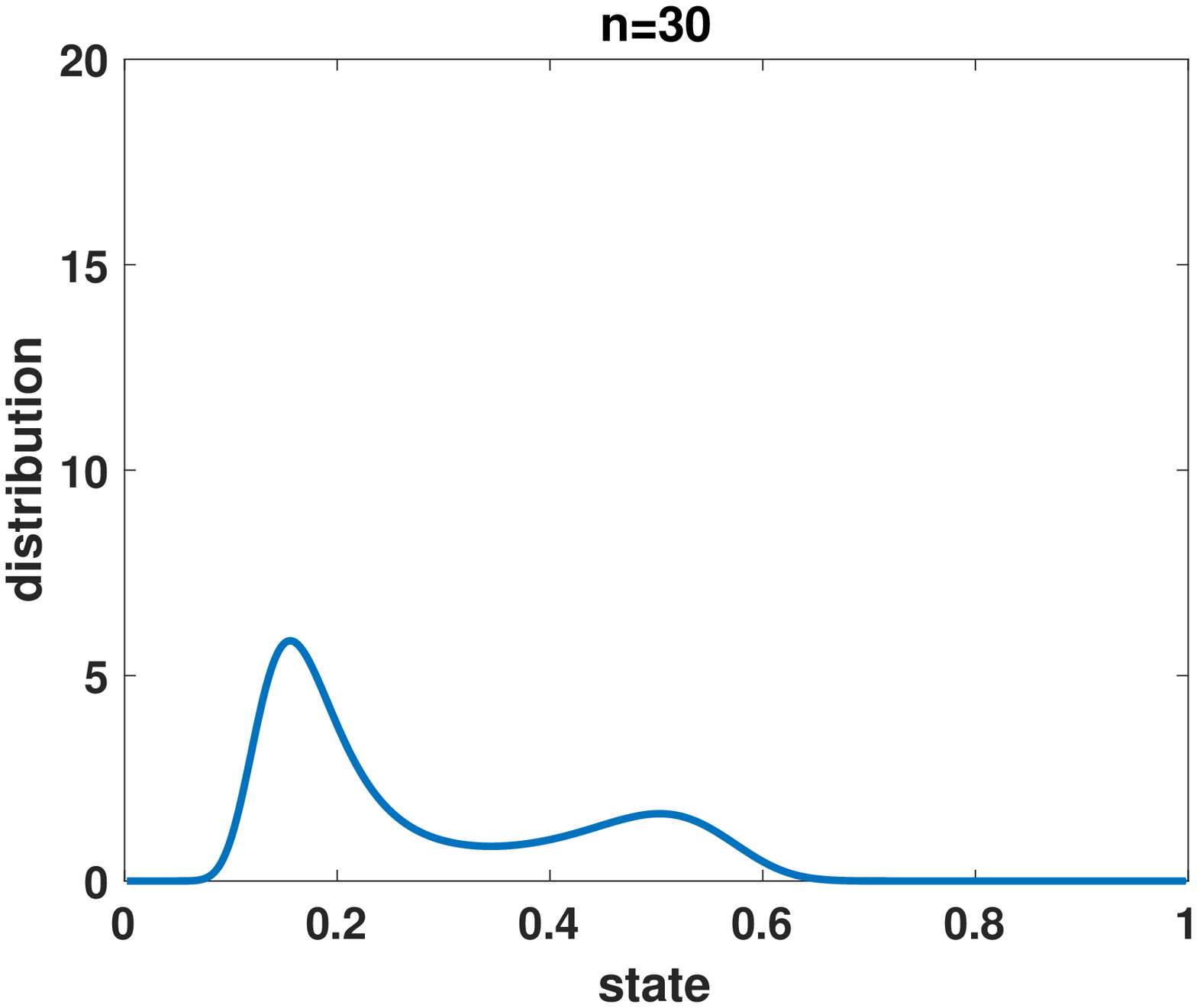}}}\quad
		\subfigure{\scalebox{0.24}{\includegraphics{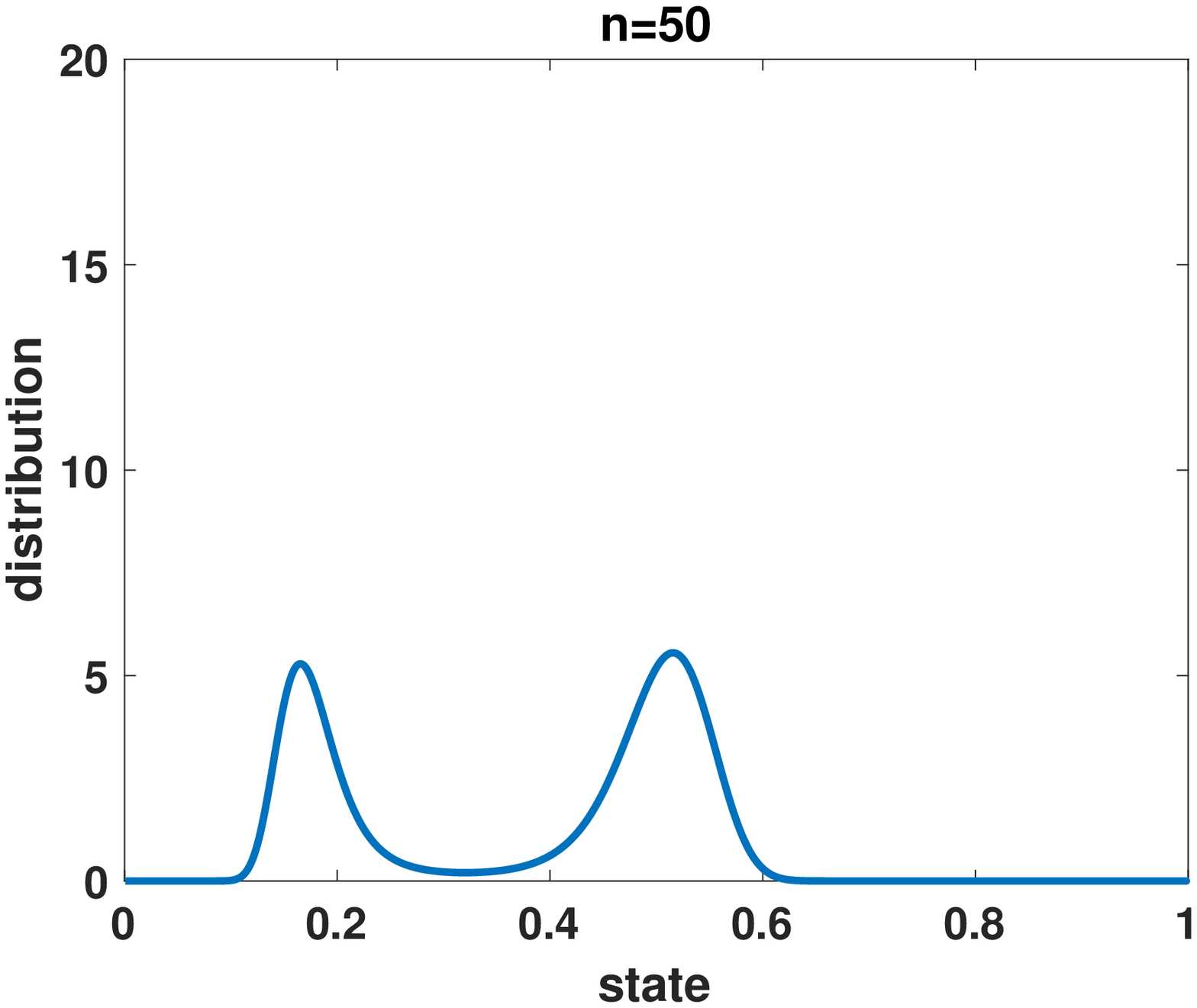}}}\quad
		\subfigure{\scalebox{0.24}{\includegraphics{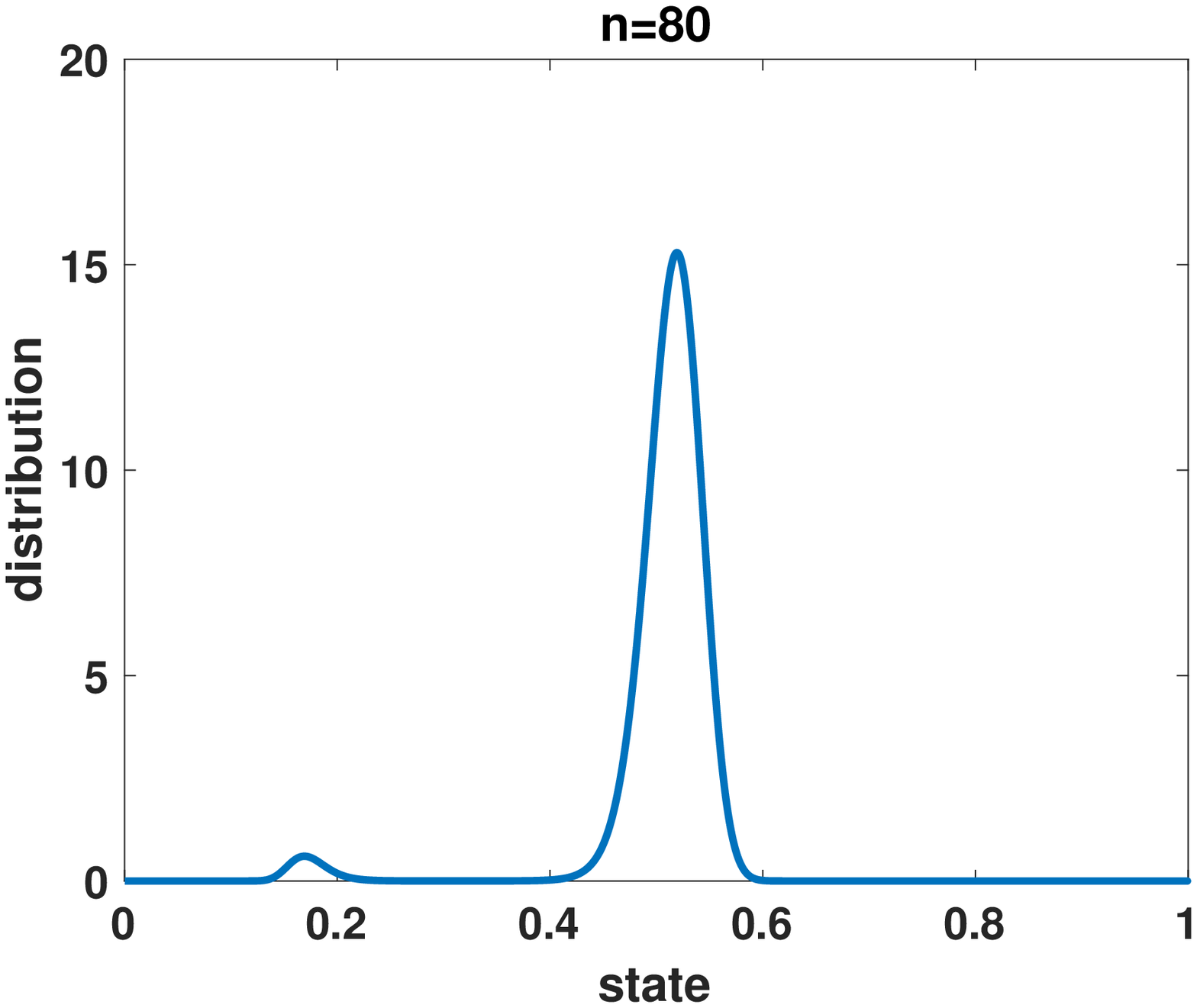}}}\quad
		\subfigure{\scalebox{0.24}{\includegraphics{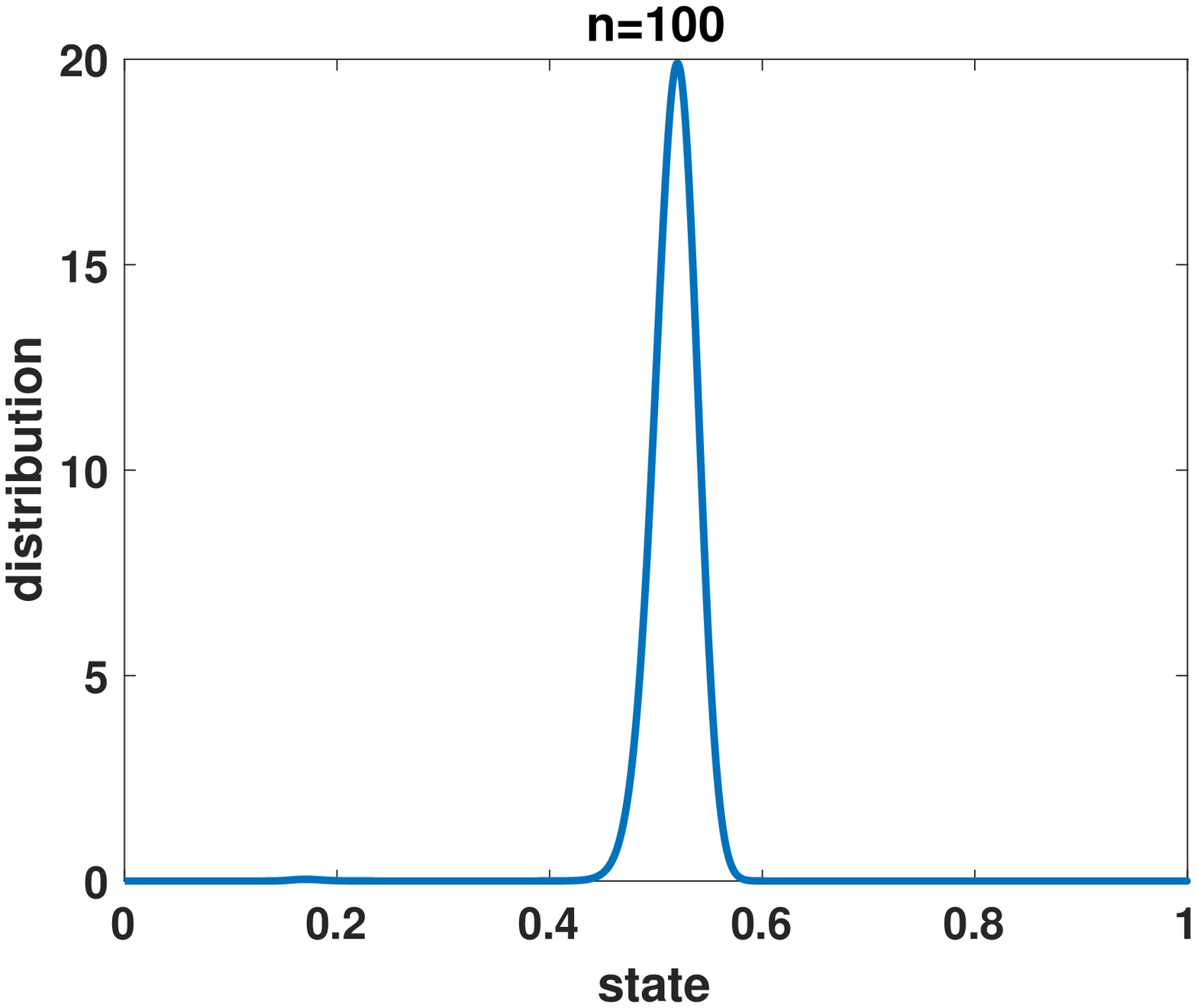}}}
 \caption{Steady state distribution of the macroscopic 
 model 
 (\ref{eq:bdb})--(\ref{eq:bdd}) for $n=30,50,80,100$ in the bistable regime (\ref{eq:bistablec}).}
  \label{fig:steady}
\end{figure}

The following result shows that the 
bimodality of the macroscopic steady state 
is connected to the roots of the cubic equation 
(\ref{eq:cubic}), which
depend only on 
$c_1$, $c_2$ and $c_3$.
We note that these roots are 
$p_1^\star \approx 0.17$,
$p_2^\star \approx 0.31$
and $p_3^\star \approx 0.52$
for the bistable regime used in 
 Figures~\ref{fig:q} and 
\ref{fig:steady}.
(We also note that the cubic equation (\ref{eq:cubic})
has a single real root of 
$p^\star \approx 0.75$
in the monostable regime (\ref{eq:monostablec}) used in the upper plots of Figure~\ref{fig:path_comparison}.)

\begin{theorem}
Consider the scalar cubic equation 
\begin{equation}
(1-p) ( c_1 + c_3 p^2) - c_2 p = 0,
\label{eq:cubic}
\end{equation}
for which any real root must necessarily lie in $(0,1)$. 
Then 
\begin{description}
    \item[(a)]
the cubic equation has one real root
    $0<p^\star<1$
if and only if    
for sufficiently large $N$
    the family of birth and death processes 
    (\ref{eq:bdb})--(\ref{eq:bdd}) with 
(\ref{eq:new_lambda_i}) and (\ref{eq:new_mu_i})
has a unimodal steady state with maximum probability at state $j$,
where $j/N = p^\star + O(1/N)$,
\item[(b)]
    the cubic equation has three real roots 
    $0< p_1^\star < p_2^\star < p_3^\star <1$
    if and only if    
for sufficiently large $N$
    the family of birth and death processes 
    (\ref{eq:bdb})--(\ref{eq:bdd}) with 
(\ref{eq:new_lambda_i}) and (\ref{eq:new_mu_i})
has a bimodal steady state with locally maximum probability at state $j_1$,
where $j_1/N = p_1^\star + O(1/N)$ and 
at 
 state $j_3$,
where $j_3/N = p_3^\star + O(1/N)$, and
with locally minimum probability at state $j_2$ where  
$j_2/N = p_2^\star + O(1/N)$.
\end{description}
\label{thm:steady}
\end{theorem}

\begin{proof}  
Let 
\begin{equation}
f(p) = \frac{(1-p)(c_1 + c_3 p^2)}{c_2},
\label{eq:fdef}
\end{equation}
so a root of (\ref{eq:cubic}) corresponds to $f(p) = p$.
Because $f(p) > 0 $ for $p  <0$ and
$f(p) < 0 $ for $p  > 1$, we see that any real root must lie in 
$(0,1)$.
We know from (\ref{eq:exact_steady}) that 
$\pi_{j+1}/\pi_j = \lambda_j/\mu_{j+1}$.
Hence, from 
(\ref{eq:new_lambda_i}),
(\ref{eq:new_mu_i})
and
(\ref{eq:fdef}), 
we see that 
that for any $j = 0,1,\ldots,N-1$, 
\begin{equation}
    \left|
       \frac{\pi_{j+1} } { \pi_j} 
       - \frac{f(j/N)}{j/N}
       \right|
       \le \frac{C}{N}, 
       \label{eq:fdiff}
\end{equation}
where $C$ is a constant that is uniform in $j$.

For case (a), suppose that the cubic has one root $p^\star$ in $(0,1)$.
Then, by construction,
\[
\frac{f(p)}{p} > 1 \text{~for~} p \in (0,p^\star) 
\quad 
\text{and} \quad 
\frac{f(p)}{p} < 1 \text{~for~} p \in (p^\star,1).
\]
Let $j^\star$ be the largest 
$j$ such that $j/N \le p^\star$. Then, for sufficiently large $N$,
we see from (\ref{eq:fdiff}) that 
\[
\frac{\pi_{j+1}}{\pi_j} > 1 \text{~for~} j < j^\star 
\quad 
\text{and} \quad 
\frac{\pi_{j+1}}{\pi_j} < 1
\text{~for~} j > j^\star.
\]
Hence,
$\pi_0 < \pi_1 < \cdots < \pi_{j^\star}$ and
$ \pi_{j^\star+1} >  \pi_{j^\star+2}> \cdots > \pi_{N}$.
So one or both of
$\pi_{j^\star}$
and 
$\pi_{j^\star+1}$
give a maximum value.

On the other hand, suppose we have a unimodal 
distribution, that is, 
$\pi_0 < \pi_1 < \cdots < \pi_{j^\star}$ and
$ \pi_{j^\star+1} >  \pi_{j^\star+2}> \cdots > \pi_{N}$.
Then, by 
(\ref{eq:fdiff}),
\[
\frac{f(p)}{p} + O(1/N) > 1 \text{~for~} p = \frac{j}{N} \text{~with~} p < p^\star  
\]
and
\[
\frac{f(p)}{p} + O(1/N) < 1 \text{~for~} p = \frac{j}{N} \text{~with~} p > p^\star.
\]
For this to be true for all sufficiently large $N$, by continuity we 
must have a unique root in $(0,1)$ at $p^\star$.

Case (b) may be proved in a similar manner.
\end{proof}

\subsection{Exit Times}

Figures~\ref{fig:q} and \ref{fig:steady} suggest that for large $n$ there is a metastability effect, where 
the process spends a long time in a  
one regime before eventually switching to another.
It is, of course, extremely challenging to 
verify such behavior by directly simulating a trajectory  \cite{Carr1995}.
To gain further understanding of this type of behavior, we will study the expected time taken to move between modes in the bimodal case of the macroscopic model. 

Let $G \in \mathbb{R}^{(N+1) \times (N+1)}$ be the
tridiagonal matrix 
\begin{equation}
\begin{bmatrix}
-\lambda_0& \lambda_0&0&\dots&0\\ \\
\mu_1 & -(\lambda_1+\mu_1)&\lambda_1 &\ddots & \vdots\\ \\
0 &\mu_2&-(\lambda_2+\mu_2)&\lambda_2 & \vdots\\ \\
\vdots &\ddots & \ddots&\ddots& \lambda_{N-1} \\ \\
0&\dots&\dots&\mu_N&-\mu_N
\end{bmatrix},
\end{equation}
and, given $0 \le i \le N$, let $Q^{[i]} \in \mathbb{R}^{N \times N}$ be the matrix obtained from $G$ by neglecting the $i$-th row and column, where we index from $0$ to $N$. Denote by $1_N$ the vector in $\mathbb{R}^N$ of all ones, and consider the 
linear system
\begin{equation}
Q^{[i]} \tau^{[i]} = -1_N.
\label{eq:exit}
\end{equation}
The component $\tau^{[i]}_j$ then 
gives the mean time spent by the process
(\ref{eq:bdb})--(\ref{eq:bdd}) 
to reach the state $i$ starting from  state $j$;
that is, 
\[
\tau^{(i)}_j = \mathbb{E}\left[ t : X(t) = i|X(0) = j\right], \quad j=0,\dots, N, \quad j\ne i.
\]

Let $\lfloor \cdot \rfloor$ denote the integer flooring operation.
In the case where the cubic 
(\ref{eq:cubic}) has three real roots,
$0 < p_1^\star < p_2^\star < p_3^\star$, 
and 
for a fixed value $n$, and hence $N$, we 
 let $N_{1}=\lfloor p_1^{\star}N \rfloor$,
 $N_{2}=\lfloor p_2^{\star}N \rfloor$
 and 
 $N_{3}=\lfloor p_3^{\star}N \rfloor$.
 Hence, by Theorem~\ref{thm:steady},
 $N_1$ and $N_3$ are 
 the state values corresponding to the 
 peaks of the bimodal density,
 and $N_2$ corresponds to the trough in between.
 In Figure~\ref{fig:exit_ratio}, on the left, 
 by repeatedly solving appropriate systems of the form (\ref{eq:exit}), we show 
 the mean time taken for the process to 
 move from $N_1$ to $N_2$ (dashed line) and from
 $N_3$ to $N_2$ (solid line), as a function of $n$.
  This characterizes for a specific system size $n$ how long 
it takes on average to switch between the two regimes. 
Note that both axes in Figure~\ref{fig:exit_ratio} are logarithmically 
scaled.
We see that as the system size increases, the time to transition between the two modes increases dramatically.
Moreover, for small $n$, the 
$N_1 \to N_2$ transition takes longer than the
$N_3 \to N_2$ transition, whereas this 
relation changes over as $n$ increases.
In Figure~\ref{fig:exit_ratio}, on the right, we show the ratio of the 
$N_1 \to N_2$ and
$N_3 \to N_2$ mean exit times.
This behavior is consistent with 
the steady state plots in Figure~\ref{fig:steady}, where
for small (respectively, large) $n$ there is more probability mass around
$N_1$ (respectively, $N_3$).

\begin{figure}[h]
  \centering
		\subfigure{\scalebox{0.24}{\includegraphics{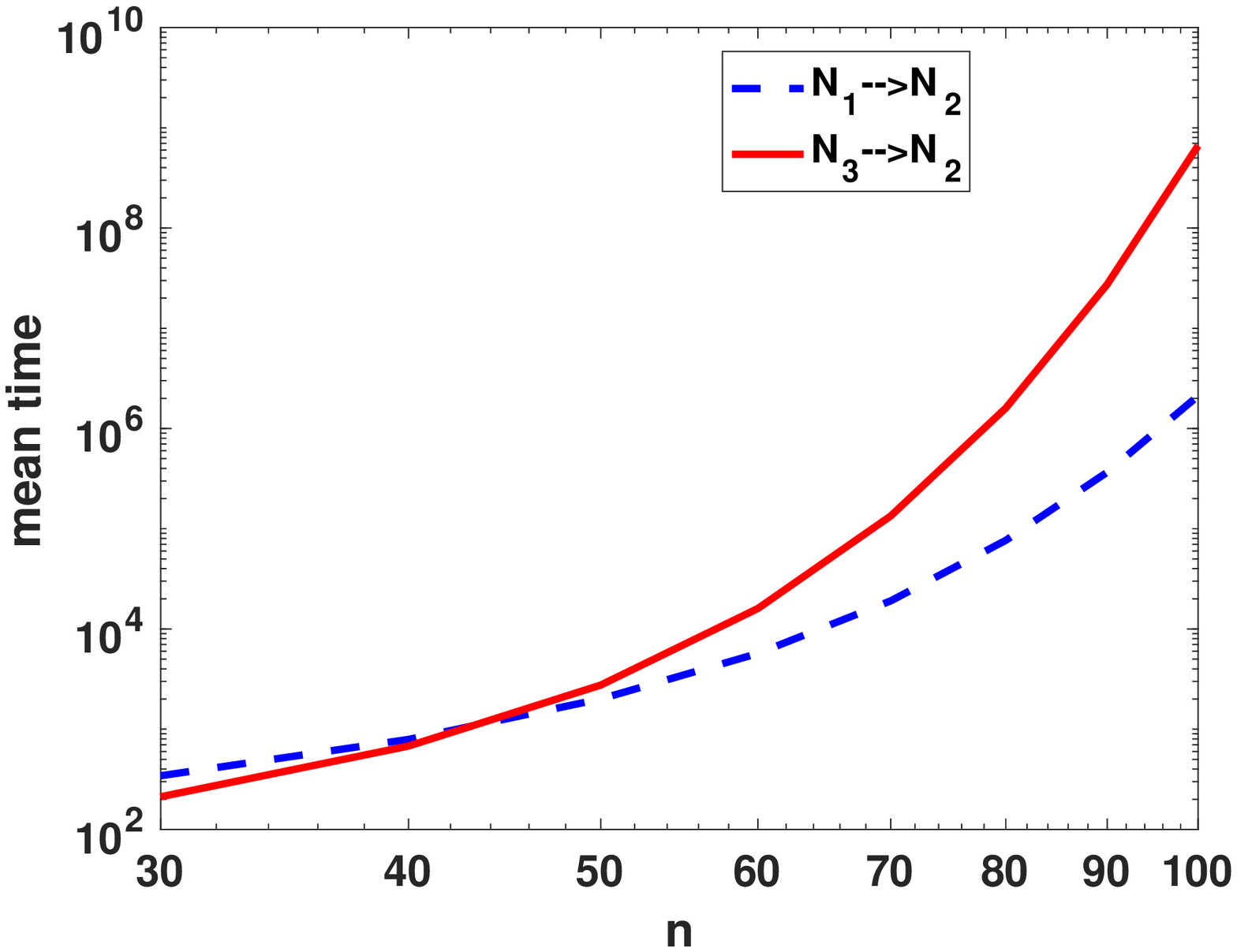}}}\quad
		\subfigure{\scalebox{0.24}{\includegraphics{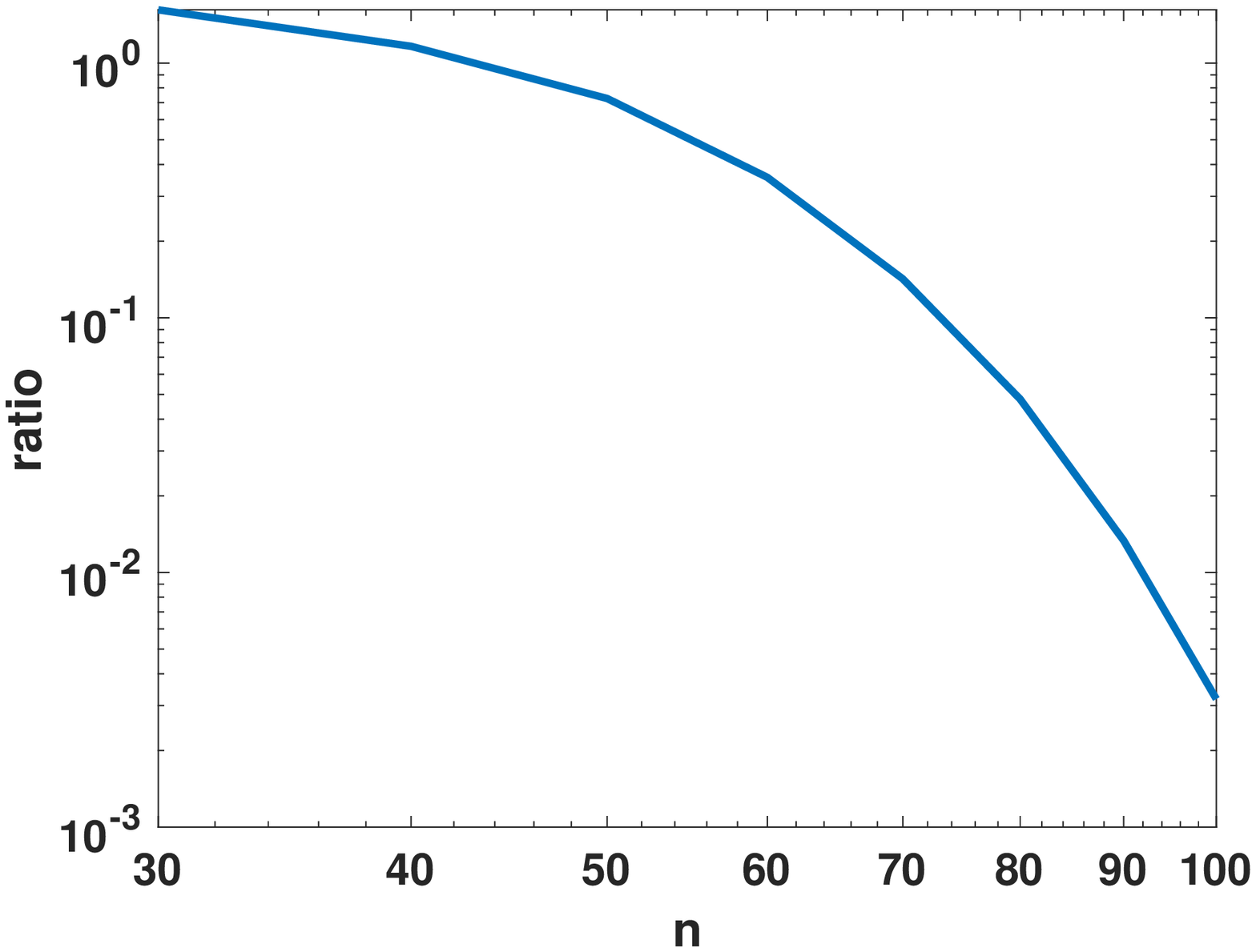}}}
 \caption{Bistable regime: mean time for the macroscale model to reach the central trough of the steady state,
$N_2$, starting from each peak, $N_1$ and $N_3$ (left). Ratio between correspondent times (right).}
  \label{fig:exit_ratio}
\end{figure}



\section{Langevin and Deterministic ODE Models}\label{sec:sde}


 
 The macroscale model 
 (\ref{eq:bdb})--(\ref{eq:bdd}) with 
(\ref{eq:new_lambda_i}) and (\ref{eq:new_mu_i})
may be viewed as 
part of a hierarchy of models
that includes a stochastic differential equation (SDE), or
Langevin process, and a deterministic reaction rate ODE. 
 We will first look at the Langevin version and test whether it retains the metastable 
 behavior observed earlier.
 
 Using standard arguments 
  \cite{G91,G2000} the 
Langevin equation may be written as  
an Ito SDE of the form 
\begin{subequations}
\label{eq:Langevin}
\begin{eqnarray}
dy(t)&=& \mu(y(t))dt+\beta(y(t))d\mathbf{W}(t), \\
\mu(y) &=& c_1\left(1-y\right)-c_2 y+c_3(1-y)y^2, \\
\beta(y) &=& \frac{1}{\sqrt{N}}[\sqrt{c_1(1-y)},-\sqrt{c_2 y},\sqrt{c_3 (1-y)y^2}],
\end{eqnarray}
\end{subequations}
where 
$\mathbf{W}(t)= [W_{1}(t),W_{2}(t),W_{3}(t)]^{T}$
is formed of three independent Brownian motions.
Here the scalar, real-valued, random variable $y(t)$ 
represents the edge density at time $t$. 
Because of the presence of the square roots 
in (\ref{eq:Langevin}c), 
the process becomes ill-defined when
$y(t)$ leaves the interval $[0,1]$.
This circumstance is common in Langevin formulations; intuitively, it arises from the fact that 
the approximations used to transfer between Poisson and  
diffusion processes break down when the molecule count for any species becomes small.
Many approaches have been proposed to overcome this 
issue, see, for example, \cite{constL19,Hybrid16,MAG21} and the references therein.
We will avoid this difficulty by focusing on 
the mean first passage time for the SDE with a reflecting boundary at one end, thereby 
considering only paths that do not leave $[0,1]$.
Letting  
\[
\sigma^{2}(y)=\frac{1}{N}\left(c_1\left(1-y\right)+c_2 y+c_3(1-y)y^2\right),
\]
we may then define the following partial differential operator
\begin{equation} \label{eq:L}
Lu:= \mu (y) \frac{du}{dy}+\frac{1}{2}\sigma^{2}(y)\frac{d^{2}u}{dy^{2}},      
\end{equation}
and the corresponding boundary value problem in the interval $(a,b)$
\begin{equation} \label{eq:mean_exit}
LT=-1,
\end{equation}
subject to one of the two following pairs of boundary conditions: 
\begin{subequations} \label{eq:bound}
\begin{eqnarray}
T'(0) &=&0, \quad T(b)=0, \\
T(a) &=&0, \quad T'(1)=0.
\end{eqnarray}
\end{subequations}
The solution $T(x)$
with boundary conditions (\ref{eq:bound}a)
gives the mean time for paths of \eqref{eq:Langevin} starting from $x$ to 
reach the value $b$ when the SDE has a reflecting boundary condition at $y(t) = 0$ \cite{Gardiner}. 
Similarly, 
with the boundary conditions (\ref{eq:bound}b),
$T(x)$ gives the mean time for the SDE paths to reach $a$ 
starting from $x$
when there is a reflecting boundary condition at $y(t) = 1$.

In the manner of 
Figure~\ref{fig:exit_ratio}, in 
Figure~\ref{fig:exit_time1} we show the 
mean times taken for the 
Langevin process to 
 reach the
 central trough of the steady 
 state distrubution (scaled to represent edge count rather 
 than edge density) from each peak.
 The curves were obtained by numerically solving the relevant boundary value problems.
 Based on separate computations, not reported here, we 
believe that it is very unlikely for an SDE path 
to leave the interval $[0,1]$ over any reasonable timescale, and hence the reflecting boundary condition, which is introduced to avoid the technical issue of 
ill-defined square roots, has little effect on the 
results.

\begin{figure}[htbp]
\centering
\scalebox{0.4}{
\includegraphics{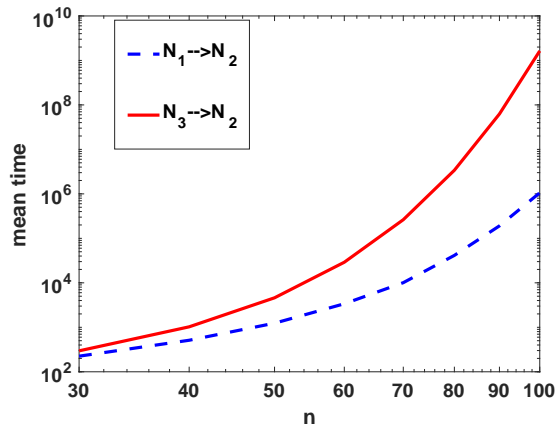}}
\caption{Analogues of the mean exit times in 
Figure~\ref{fig:exit_ratio} for the 
Langevin equation \eqref{eq:Langevin}.
}
\label{fig:exit_time1}
\end{figure}

For further comparison, in Figure~\ref{fig:exit_time_comp} 
we superimpose the mean exit times
from Figure~\ref{fig:exit_ratio} (on the left)
for the macroscale model and
the mean first passage times 
from Figure~\ref{fig:exit_time1} for the 
Langevin process.
Here, we have plotted 
$n^2$ on the horizontal axis against the logarithm of the mean time on the vertical axis.
We see that 
the SDE captures the extreme rate observed for the macroscale model, and 
for 
both models the mean time appears to scale like
$\exp(n^2)$, in line with the reaction rate theory of Kramers
\cite{RevModPhys.62.251,Pavliotisbook}.
It would, of course, be of interest to pursue 
the rigorous analysis 
of 
\eqref{eq:Langevin} with respect to long term behavior 
as the system size increases.

\begin{figure}[htbp]
\centering
\scalebox{0.4}{
\includegraphics{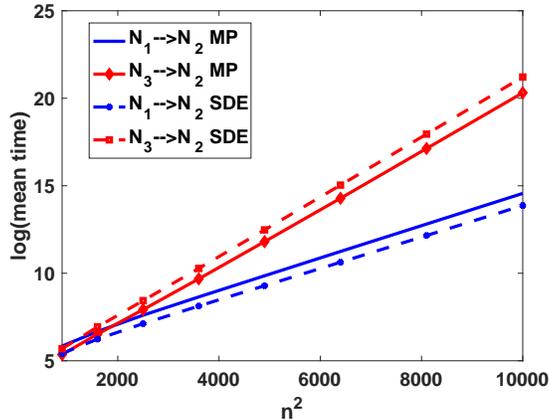}}
\caption{Comparison between mean exit times for the
SDE
\eqref{eq:Langevin} (as in Figure~\ref{fig:exit_time1}) and the macroscale Markov process
 (\ref{eq:bdb})--(\ref{eq:bdd}) with 
(\ref{eq:new_lambda_i}) and (\ref{eq:new_mu_i})
(as in Figure~\ref{fig:exit_ratio}).
}
\label{fig:exit_time_comp}
\end{figure}

\subsection{Deterministic Model}

The reaction rate ODE
arising from 
\eqref{eq:Langevin} takes the form 
$dy/dt = \mu(y(t))$. 
We note that $\mu$ coincides with the 
cubic in (\ref{eq:cubic}), and 
hence the edge densities 
$p^\star$ in part (a) of Theorem~\ref{thm:steady}
and 
$p_1^\star,p_2^\star,p_3^\star$ in part (b) of Theorem~\ref{thm:steady} coincide with fixed points of this
ODE. 
We also note that the deterministic mean field
from \cite{GHP12} is essentially an Euler discretization of this ODE with unit stepsize.
In the upper plots of Figure~\ref{fig:path_comparison}
we see that, in the monostable regime, the deterministic 
mean field approximation captures the long term behavior 
of the 
macro-level edge density.
In the lower plots, corresponding to the  
bistable regime, 
the deterministic approximation is not able to switch, and hence settles on one of the two stable fixed points,
$p_1^\star$ or $p_3^\star$, depending on the initial condition.

\section{Discussion}\label{sec:disc}

Our aim in this work was to 
perform multiscale stochastic modeling and analysis to add insights to the bistability effect 
observed in \cite{GHP12}.
A major objective was to develop
simplified 
stochastic models that 
(a) reproduce the switching and metastability behavior observed for the full microscale model,  (b) make long-term simulation feasible, and (c) offer the potential for  
rigorous analysis.
A key step in this work was to 
set up a suitable scaling
for the triadic closure
reaction rate as a function of system size, i.e., the total number of possible edges, $N$.
Returning to the original microscale model 
 (\ref{eq:OE})--(\ref{eq:EEO}), we note that $\wc3$ is the rate constant for 
 each individual triadic closure reaction.
 So the chance of 
 reaction
  (\ref{eq:EEO}) happening is 
  $\wc3$ if
  each of   
$ O_{ij} $,   $ E_{jk} $ and $ E_{ik} $ exist,
and zero otherwise.
Based on the way our macroscale 
approximations arose, 
it was then natural for $\wc3$ to 
have the scaling 
$\wc3 \propto 1/\sqrt{N}$
in (\ref{eq:c3scale}).  
Here, \emph{the chance of any individual triadic closure event taking place decreases as the system size grows}.
This scaling is necessary if we wish to be in a regime  where 
spontaneous birth, spontaneous death and 
triadic closure all coexist for large system size.
In particular, keeping $\wc3$ constant as $N$ grows would lead to a system where triadic closure
dominates, in general.

There are several directions in which this work could be pursued.
For example, 
we note that although 
online 
social networking 
has the potential 
to increase 
the number of friendships an individual may form,
there is evidence of an upper limit of
 around 150 in practice 
   \cite{B06,D16}.
   This ``Dunbar  number''
effect can be explained 
  in part by cognitive constraints and in part by the time costs of maintaining relationships 
\cite{Dunbar,D16}.
In our model, in a bistable regime  
the low and high 
connectivity levels  
($p_1^\star N$ and $p_2^\star N$ in Figure~\ref{fig:steady})
lead to nodal degrees that become arbitrarily large
in terms of
$N$.
Hence, in the case of very large networks it would be reasonable to 
investigate models that 
incorporate some sort of saturation or carrying capacity effect in order to limit the maximum number of friendships a person may maintain. 

Furthermore, it is of course possible to move from triadic closure to more general mechanisms that encourage cliques
\cite{clique2015} or other motifs to 
form in a network.
Here new challenges include
the incorporation of appropriate higher order 
matrix-level nonlinearities in expressions such as 
(\ref{eq:atriadic}), and the justification of
mean field approximations across longer-range interactions.
Finally, it would be of interest to develop calibration methods along the lines of \cite{MH20} in order to  fit parameters to real data and  search for bistability effects.

\end{document}